\documentclass[reqno,11pt]{amsart}
\usepackage{amssymb, amsmath, amsthm}
\usepackage[colorlinks=true,linkcolor=blue,citecolor=red]{hyperref}
\usepackage{color}
\usepackage{epsfig}
\usepackage{amsmath}
\usepackage{amssymb}
\usepackage{amscd}
\usepackage{graphicx}
\usepackage{mathrsfs}
\usepackage{enumerate}
\usepackage{tikz}
\usepackage{epstopdf}
\usepackage[T1]{fontenc}
\numberwithin{equation}{section}
\newtheorem{lemma}{Lemma}[section]
   \newtheorem{thm}{Theorem}[section]
   
   \newtheorem{prop}{Proposition}[section]
   \newtheorem{rem}[thm]{Remark}
\newtheorem{assum}[thm]{Assumption}
\newtheorem{RHP}[thm]{Riemann-Hilbert Problem}

\newtheorem{cor}{Corollary}[section]

\usepackage{epsfig}

\numberwithin{equation}{section}
\numberwithin{prop}{section}
\numberwithin{lemma}{section}
\numberwithin{re}{section}
\numberwithin{coro}{section}

\newcommand{\R}{\mathbb{R}}



\begin{document}

\title[Long time asymptotic for the NNLS equation with finite density type initial data]{Long-time asymptotic behavior of the nonlocal nonlinear Schr\"{o}dinger equation with finite density type initial data}


\author[Tian]{Shou-Fu Tian$^{*}$}
\address{ Shou-Fu Tian (Corresponding author), Zhi-Qiang Li and Jin-Jie Yang \newline
School of Mathematics, China University of Mining and Technology, Xuzhou 221116, People's Republic of China}
\thanks{$^{*}$Corresponding author(sftian@cumt.edu.cn, shoufu2006@126.com).
}
\email{sftian@cumt.edu.cn, shoufu2006@126.com (S.F. Tian)}
\author[Li]{Zhi-Qiang Li}
\author[Yang]{Jin-Jie Yang}

%

\begin{abstract}
{In this work, we employ the $\bar{\partial}$-steepest descent method  to investigate
the Cauchy problem of the nonlocal nonlinear Schr\"{o}dinger (NNLS) equation with finite density type initial conditions in weighted Sobolev space $\mathcal{H}(\mathbb{R})$.
Based on the Lax spectrum problem, a Riemann-Hilbert problem corresponding to the original problem is constructed to give  the solution of the NNLS equation with the finite density type initial boundary value condition.
By developing  the $\bar{\partial}$-generalization of Deift-Zhou nonlinear steepest descent method, we derive  the leading order approximation to the solution $q(x,t)$ in soliton region of space-time, $\left(\frac{x}{2t}\right)=\xi$ for any fixed $\xi=\in (1,K)$($K$ is a sufficiently large real
constant), and give bounds for the error decaying as $|t|\rightarrow\infty$.
Based on the resulting asymptotic behavior, the asymptotic approximation of the NNLS equation is characterized with the soliton term confirmed by $N(\Lambda)$-soliton on discrete spectrum and  the $t^{-\frac{1}{2}}$ order term on continuous spectrum  with residual error up to $O(t^{-\frac{3}{4}})$.
}
\end{abstract}

\maketitle

\tableofcontents

\section{Introduction}

The nonlinear Schr\"{o}dinger (NLS) equation is a classical integrable  model with Lax pairs, infinite conservation laws and Hamiltonian structures
\begin{align}\label{e1}
iq_t(x,t)\pm q_{xx}(x,t)+2|q(x,t)|^2q(x,t)=0,
\end{align}
which can be used to describe nonlinear optics, plasma and other phenomena \cite{NLS-optic}.
Due to its significant physical meaning and rich mathematical structures, many researchers and scholars are devoted their effort to investigate various aspects of the NLS equation \eqref{e1} and its extensions \cite{Minakov-JMP,Prinari-PD,Tian-PAMS,Tian-JDE,Tian-PA,Wangds-2019-JDE}.

With the gradual deepening of the research, the research on the NLS equation is also more in-depth. In 1998, Bender et al. \cite{PT-symmetry} introduced parity-time ($\mathcal{PT}$) symmetry in the generalized Hamiltonians. In 2008,  parity-time ($\mathcal{PT}$) symmetry was introduced  in the NLS equation \cite{PT-symmetry-NLS}.

In 2013, Ablowitz et al. introduced   the $\mathcal{PT}$ symmetry  to the first one of the well-known AKNS system to present a nonlocal nonlinear Schr\"{o}dinger (NNLS) equation \cite{PRL-2013}
\begin{align}\label{e2}
iq_{t}(x,t)-q_{xx}(x,t)+2q^{2}(x,t)q^{*}(-x,t)=0,
\end{align}
where $q(x,t)$ is complex function, and the superscript $*$ means the complex conjugate.
It is noted that the NNLS equation and the unconventional system of Landau-Lifshitz equations are equivalent under the sense of gauge transformation \cite{NNLS-LL-equ}.
Moreover, it has been pointed out that the NNLS equation plays a more and more important role in  the theoretical study of mathematical physics and applications in the fields of nonlinear science \cite{NNLS-Phy-Yang}.

Therefore, more and more attentions are paid to the NNLS equation, for instance, a detailed study of the inverse scattering transform (IST) of the NNLS equation is carried out \cite{NNLS-Ablowitz-IST}. In 2018, Ablowitz et al. \cite{NNLS-Ablowitz-IST-nonzero} studied the NNLS equation with nonzero boundary conditions by using IST method.
Additionally, utilizing  Hirota's  bilinear method and the KP hierarchy reduction method, Feng et al.\cite{NNLS-Feng-Bilinear-zero-nonzero}  investigated the general soliton solutions of the NNLS equation with zero and nonzero boundary conditions. In 2019, Rybalko and Shepelsky paid attention to the long time asymptotics of the NNLS equation with decaying  boundary conditions \cite{NNLS-long-time-JMP}.
Furthermore, they were devoted their effort to investigate long-time asymptotics for the NNLS equation with step-like initial data
and  a family of step-like initial data \cite{NNLS-long-time-JDE,NNLS-long-time-CMP}.
Moreover, the mixed soliton solutions of the defocusing NNLS equation were also studied by using  Darboux transformation and  developing the asymptotic analysis method \cite{NNLS-DT-PD}.
In 2021, shifted nonlocal reduction formulae was applied to study the soliton solutions of the NNLS equation \cite{NNLS-G-solition}. What's more, lots of  great work on NNLS equation have been reported \cite{NNLS-Chen-Studies,NNLS-Fokas-Nonlinearity,NNLS-Gerdjikov-JMP, NNLS-Xu-PRE,NNLS-Ma-Studies,NNLS-Ma-PAMS,NNLS-Yan-AML}.

However, to the best of our knowledge, the research using the  $\bar{\partial}$-steepest descent method for the NNLS equation has not been reported yet.
In this work, we investigate the long-time asymptotic behavior for the Cauchy problem
of NNLS equation with finite density type initial data
\begin{align}\label{BC}
q(x,0)=q_{0}(x),~\lim_{x\rightarrow\pm\infty}q_{0}(x)=1.
\end{align}

In 1974, Manakov first paid attention to the long time asymptotic behavior of the  nonlinear evolution equations \cite{Manakov-1974}. Later, inspired of  Manakov, Zakharov et al. made  their efforts to further study long time asymptotic behavior of nonlinear evolution equations \cite{Zakharov-1976}.

In 1993, a crucial work was reported for the research on the long time asymptotic behavior of the  nonlinear evolution equations. That is the nonlinear steepest descent method was developed
to systematically study the long time asymptotic behavior of nonlinear evolution equationsby Defit and Zhou \cite{Deift-1993}. Since then, based on the work of Defit and Zhou, many scholars made their contribution to the development of the nonlinear steepest descent method. Therefore, with the development of the nonlinear steepest descent method,  the error term became better. For instance, with enough smooth  and decaying  fast  initial value, the error term could  reach $O(\frac{\log t}{t})$ which were shown in \cite{Deift-1994-1,Deift-1994-2}. With further research, the error term became $O(t^{-(\frac{1}{2}+\iota)})$ for any $0<\iota<\frac{1}{4}$  with  weighted Sobolev initial data \cite{Deift-2003}. As the nonlinear steepest descent method  matures and the accuracy gradually increases,   it has been applied widespread to solve integrable system, including KdV equation, Fokas-Lenells equation, short pulse equation, complex short pulse equation, Camassa-Holm equation and so on \cite{Monvel-CH-Long-time,Grunert-KdV-Long-time,Xu-SP-JDE,Xu-CSP-JDE}.

In recent years, combined nonlinear steepest descent with $\bar{\partial}$-problem, McLaughlin and Miller \cite{McLaughlin-1, McLaughlin-2}   developed a $\bar{\partial}$-steepest descent method to investigate the asymptotic of orthogonal polynomials. Then, inspired by their work,  Dieng et al. \cite{Dieng-2008} studied the long-time asymptotics of the NLS equation with finite mass initial data  by developing the $\bar{\partial}$-steepest descent method. And Cuccagna et al. \cite{Cuccagna-2016} investigated the  asymptotic stability of $N$-solitons of the defocusing nonlinear Schr\"{o}dinger equation with finite density initial data successfully with the help of the $\bar{\partial}$-steepest descent method.
Since then, the  $\bar{\partial}$-steepest descent method  has been applied widespread \cite{AIHP,Faneg-2,Miller-2,Jenkins,Jenkins2,Li-cgNLS,Fan-mKdV-Dbar-1, Fan-SP-Dbar,mCH-Fan-AM,Yang-Hirota}.
The reason why the   $\bar{\partial}$-steepest descent method is applied widespread is that it possesses the advantages that the nonlinear steepest descent method does not have. For example, compared with the nonlinear steepest descent method, in the process of utilizing the   $\bar{\partial}$-steepest descent method,
the delicate estimates involving $L^{p}$ estimates of Cauchy projection operators can be avoided. Moreover, the accuracy of the long-time asymptotic results  is also improved. For instance, in \cite{Dieng-2008}, with the weighted Sobolev space initial value,  the error term can reach $O(t^{-\frac{3}{4}})$ which is a great improvement from $O(\frac{\log t}{t})$.

In this   present work, we study the  asymptotic behavior and long time asymptotic behavior for the initial value
problem of the NNLS equation \eqref{e2}-\eqref{BC}. In our work,  we ease the restrictions on the Schwartz initial data by allowing the weighted
Sobolev initial data $q_{0}(x)\in \mathcal{H}(\mathbb{R})$ and also allowing presence of discrete spectrum.

\noindent \textbf{Plan of the proof}

We give the proof of the theorem \ref{Thm-1}
through extending  the inverse scattering transform, the Dbar techniques and the Riemann-Hilbert (RH) problem to the NNLS equation \eqref{e2}.

In Section 2, taking into account  the Lax spectrum problem of NNLS equation \eqref{e2}, we introduce eigenfunctions $\mu_{\pm}$ depending on the initial data $q_{0}(x)$ to deform  the Lax pair and analyse its properties. Moreover, the corresponding scattering matrix $S(z)$ is further analyzed.

In Section 3, with the assumption that the scattering data in $S(z)$
have no spectral singularities on real axis and have finite simple zeros on complex plane, i.e., Assumption \ref{assum},  a RH problem for $M(z)$ corresponding to the Lax spectrum problem is constructed with weighted Sobolev initial data. Meanwhile,  the reflection coefficient $r(z)$ and transmission coefficient $\check r(z)$ defined by \eqref{r-def} is of good properties with weighted Sobolev initial data condition, i.e., $r(z), \check r(z)\in H^{1,1}(\mathbb{R})$.
In Section 4, according to the oscillation term $\exp(2it\theta(z))$ in RH problem \ref{RH-1} and applying its sign distribution diagram shown in Fig. \ref{fig-1},  we introduce a set of conjugations
and interpolations transformation to deform the $M(z)$  into $M^{(1)}(z)$ which is  a standard RH problem.  Then, we make the continuous extension of the jump matrix off the real axis by introducing a matrix function $R^{(2)}(z)$ and get a mixed $\bar{\partial}$-RH problem in Section 5. The purpose
is to make each oscillation term in the triangular decomposition of the jump matrix bounded in a given region.
The main operating process for studying  the long time asymptotic behavior of the NNLS equation \eqref{e2}  can be expressed as
$M(z)\leftrightarrows M^{(1)}(z)\leftrightarrows M^{(2)}(z)\leftrightarrows M^{(3)}(z)\leftrightarrows E(z)$. The matrix $M^{(2)}(z)$ is the  mixed $\bar{\partial}$-RH problem by introducing matrix function $R^{(2)}(z)$.  Based on the $\bar\partial$ derivative of $R^{(2)}(z)$,  we decompose the mixed $\bar{\partial}$-RH problem into two parts that are a model RH problem \ref{RH-rhp} for $M^{(2)}_{RHP}(z)$ with $\bar{\partial}R^{(2)}=0$ and a pure $\bar{\partial}$-RH problem \ref{RH-4} for $M^{(3)}(z)$ with $\bar{\partial}R^{(2)}\neq0$.

The model RH problem $M^{(2)}_{RHP}$ is solved via an outer model $M^{out}(z)$ for the soliton part and inner model $M^{loc,k}$ near the phase point $\xi_1$ and $\xi_2$ which can be solved by matching  parabolic cylinder model problem respectively. Also, the error function $E(z)$ with a small-norm RH problem is achieved.
Moreover, the  existence and asymptotic behavior of the pure $\bar{\partial}$-RH problem for $M^{(3)}$ is analysed. The specific implementation is reflected in Sections 6 to 8.

\section{Spectral analysis}

In this section, we will give some established results of the  nonlocal NLS equation.   These results are known and the  interested reader can find  pedagogical  and detailed treatments in  \cite{Nonl-2016}.

The NNLS equation \eqref{e2} is completely integrable and possesses the Lax pair  representation
\begin{align}\label{e3}
\Phi_x(x,t;k)=X\Phi(x,t;k),~~
\Phi_{t}(x,t;k)=T\Phi(x,t;k),
\end{align}
where $X=-ik\sigma_3+Q$ and $T=2ik^2\sigma_3+i\sigma_3(Q^2-q_{0}^2)-2kQ-i\sigma_3Q_{x}$ with
$$\sigma_{3}=\left(
               \begin{array}{cc}
                 1 & 0 \\
                 0 & -1 \\
               \end{array}
             \right),~~Q=\left(
               \begin{array}{cc}
                 0 & q(x,t) \\
                 q^*(-x,t) & 0 \\
               \end{array}
             \right).
$$
In order to simplify the notation, here we give the expression of the  standard Pauli matrices, i.e.,
\begin{align*}
 \sigma_{1}=\left(
                            \begin{array}{cc}
                              0 & 1 \\
                              1 & 0 \\
                            \end{array}
                          \right),~~\sigma_{2}=\left(
                            \begin{array}{cc}
                              0 & -i \\
                              i & 0 \\
                            \end{array}
                          \right),~~\sigma_{3}=\left(
                            \begin{array}{cc}
                              1 & 0 \\
                              0 & -1 \\
                            \end{array}
                          \right).
\end{align*}

Taking  the boundary conditions \eqref{BC} in Lax pair \eqref{e3}, the asymptotic spectral problem is expressed as
\begin{align}\label{e4}
\psi_{x}(x,t;k)=X_{\pm}\psi(x,t;k),~\psi_{t}(x,t;k)=T_{\pm}\psi(x,t;k),
\end{align}
where $$X_{\pm}=-ik\sigma_3+Q_{\pm},~T_{\pm}=-2kX_{\pm},$$ with
$$Q_{\pm}=\left(
            \begin{array}{cc}
              0 & 1 \\
               1 & 0 \\
            \end{array}
          \right).
$$
A direct calculation shows that the eigenvalues of matrix $X_{\pm}$ are $\pm i\lambda$ where
\begin{align}\label{e5}
\lambda^2=k^2-1.
\end{align}
It is obvious that Eq.\eqref{e5} is a multi-valued function which will lead to the analysis more difficult. Therefore, to avoid the possible difficulties,  it is necessary to  introduce uniformization variable to simplify the analysis
\begin{align}\label{e6}
z=k+\lambda,
\end{align}
from which we obtain the single-valued functions
\begin{align}\label{e7}
\lambda(z)=\frac{1}{2}(z-z^{-1}),~~k(z)=\frac{1}{2}(z+z^{-1}).
\end{align}

Since the relationship between $X_{\pm}$ and $T_{\pm}$ are that $T_{\pm}=-2k X_{\pm}$,  there exist   invertible matrices
\begin{align}\label{e8}
Y_{\pm}=\left(
          \begin{array}{cc}
            1 & -\frac{i}{z} \\
            \frac{i}{z} & 1 \\
          \end{array}
        \right),
\end{align}
which make the two matrices $T_{\pm}$ and $X_{\pm}$ diagonalize at the same time.
Then, according to the asymptotic spectral problem \eqref{e4}, the Jost solutions can be constructed as
\begin{align}\label{e9}
\psi_{\pm}=Y_{\pm}e^{-i\lambda\sigma_{3}x+2i\lambda k\sigma_3t}\triangleq Y_{\pm}e^{-it\theta\sigma_3},
\end{align}
where $\theta:=\theta(x,t;z)=\lambda(\frac{x}{t}-2k)$. In order to eliminate the oscillating term, we introduce the following  modified functions
\begin{align}\label{e10}
\mu_{\pm}(x,t;z)=\psi_{\pm}e^{it\theta\sigma_3}\rightarrow Y_{\pm},~~x\rightarrow\pm\infty.
\end{align}
Then,  the equivalent Lax pair with respect to the modified function $\mu:=\mu(x,t;z)$  is expressed  in the form of Lie brackets
\begin{align}\label{Lax-mu}
\begin{split}
(Y_{\pm}^{-1}\mu_{\pm})_x-i\lambda[Y_{\pm}^{-1}\mu_{\pm},\sigma_3]=Y_{\pm}^{-1}\Delta X_{\pm}\mu_{\pm},\\
(Y_{\pm}^{-1}\mu_{\pm})_x-2ik\lambda[Y_{\pm}^{-1}\mu_{\pm},\sigma_3]=Y_{\pm}^{-1}\Delta T_{\pm}\mu_{\pm}
\end{split}
\end{align}
where $\Delta X_{\pm}=X-X_\pm$ and $\Delta T_{\pm}=T-T_\pm$.
From \eqref{Lax-mu},  we know that the Lax pair of the modified functions can be written in a fully differential form, so the modified eigenfunctions $\mu_{\pm}$ can be uniquely expressed as the following integral form
\begin{align}\label{e11}
\begin{split}
\mu_{-}(x,t;z)=Y_{-}+\int_{-\infty}^{x}Y_{-}e^{-i\lambda(x-y)\sigma_{3}}Y_{-}^{-1}\Delta Q_{-}(y,t)\mu_{-}(y,t;z)e^{i\lambda(x-y)\sigma_{3}}\, dy,\\
\mu_{+}(x,t;z)=Y_{+}-\int_{x}^{\infty}Y_{+}e^{-i\lambda(x-y)\sigma_{3}}Y_{+}^{-1}\Delta Q_{+}(y,t)\mu_{+}(y,t;z)e^{i\lambda(x-y)\sigma_{3}}\, dy,
\end{split}
\end{align}
where $[A,B]=AB-BA$  and $\Delta Q_{\pm}=Q-Q_{\pm}$. As shown in \cite{Biondini-JMP2014}, it can be proved that the analytical properties of the modified eigenfunctions are presented in the following Propositions.

We state  the following elementary proposition without proof \cite{Cuccagna-2016,Fan-mKdV-Dbar-1}.

\begin{prop}\label{pp1}
If $q(x,t)-q_{\pm}\in L^{1}(R)$, the modified eigenfunctions $\mu_{\pm}(x,t;z)$ can be analytically extended onto the corresponding regions of the $z$-plane, that is
\begin{align}\label{e12}
\mu_{-,1}(x,t;z),~ \mu_{+,2}(x,t;z)~\text{are analytic in}~ \mathbb{C}^{+}=\left\{z\in \mathbb{C}:Im~z>0\right\},\\
\mu_{-,2}(x,t;z),~ \mu_{+,1}(x,t;z)~ \text{are analytic in}~ \mathbb{C}^{-}=\left\{z\in \mathbb{C}:Im~z<0\right\},
\end{align}
where $\mu_{\pm,j}(x,t;z)$ $(j=1,2)$ represent the $j$-th column of the eigenfunctions $u_{\pm}(x,t;z)$.
\end{prop}

As usual, $\psi_{-}$ and $\psi_+$ are the fundamental matrix solutions of the scattering problem, so there is a scattering matrix $S(z)$ that satisfies the following scattering relationship for $z\in\Sigma^{\circ}=\mathbb{R}\setminus\{0,\pm1\}$
\begin{align}\label{e13}
\psi_{+}(x,t;z)=\psi_{-}(x,t;z)S(z).
\end{align}
By applying  \eqref{e10}, Eq.\eqref{e13} is transformed into
\begin{align}\label{e14}
\mu_{+}(x,t;z)=\mu_{-}e^{-it\theta\sigma_3}(x,t;z)S(z)e^{it\theta\sigma_3}.
\end{align}
Furthermore, by expanding the above expression \eqref{e14}, we derive the  scattering reflection  as
\begin{align}\label{e15}
\begin{split}
s_{11}(z)=\frac{\det(\psi_{+,1},\psi_{-,2})}{1-z^{-2}},~~s_{22}(z)=\frac{\det(\psi_{-,1},\psi_{+,2})}{1-z^{-2}},\\
s_{12}(z)=\frac{\det(\psi_{+,2},\psi_{-,2})}{1-z^{-2}},~~s_{21}(z)=\frac{\det(\psi_{-,1},\psi_{+,1})}{1-z^{-2}}.
\end{split}
\end{align}
Combined with Proposition \ref{pp1}, the analytical properties of scattering coefficient can be presented in the following Corollary.
\begin{cor}
If $q(x,t)-q_{\pm}\in L^{1}(R)$, then $s_{11}(z) $ is analytic in $\mathbb{C}^{-}$, and  $s_{22}(z) $ is analytic in $\mathbb{C}^{+}$. Generally, $s_{12}(z)$ and $s_{21}(z)$ are defined only for $z\in\mathbb{R}$.
\end{cor}

Next,  we give the symmetry properties of the eigenfunctions $\psi_{\pm}$. By simple verification, we have the following lemma.

\begin{lemma}\label{lem1}
The matrix $X(x,t;z)$ follows the symmetries:
\begin{enumerate}[(1)]
\item The first symmetry:
\begin{align}\label{e16}
X(x,t;z)=\sigma X^*(x,t;-z^*)\sigma,~~\sigma=\left(
                                                            \begin{array}{cc}
                                                              0 & 1 \\
                                                              -1 & 0 \\
                                                            \end{array}
                                                          \right).
\end{align}
\item The second symmetry:
\begin{align}\label{e17}
X(x,t;z)=X^*(x,t;z^{-1}), ~~ k(z)=k(z^{-1}).
\end{align}
\end{enumerate}
\end{lemma}
Using Lemma \ref{lem1}, the symmetry properties satisfied by the eigenfunctions and the scattering matrix can be shown in the following proposition.
\begin{prop}\label{pp2}
The eigenfunctions $\mu_{\pm}$ and the scattering matrix $S(z)$ satisfy that
\begin{enumerate}[(1)]
\item The first symmetry:
\begin{align}\label{e18}
\mu_{\pm}(x,t;z)=-\sigma \mu^*_{\mp}(-x,t;-z^*)\sigma,~~S(z)=\sigma S^*(-z^*)^{-1}\sigma.
\end{align}
\item The second symmetry:
\begin{align}\label{e19}
\mu_{\pm}(x,t;z)=-\frac{i}{z}\mu_{\pm}(x,t;z^{-1})\sigma_3Q_{\pm}, ~~ S(z)=(\sigma_3Q_{-})^{-1}S(z^{-1})\sigma_3Q_{+}.
\end{align}
\end{enumerate}
\end{prop}
\begin{rem}
In fact, since $Q_{\pm}=\left(
            \begin{array}{cc}
              0 & 1 \\
               1 & 0 \\
            \end{array}
          \right)
$, the second symmetry can be rewritten as
\begin{align}\label{e19}
\mu_{\pm}(x,t;z)=\frac{i}{z}\mu_{\pm}(x,t;z^{-1})\sigma, ~~ S(z)=-\sigma S(z^{-1})\sigma.
\end{align}
\end{rem}
\begin{prop}\label{pp3}
For $q-q_{\pm}\in L^{1}(\mathbb{R})$, then as $z\rightarrow\infty$ with $Imz>0$, the eigenfunctions $\mu_{\pm}$ obey that
\begin{align}\label{e20}
\mu_{-,1}(x;z)=e_{1}+\frac{1}{z}\left(
                                  \begin{array}{c}
                                    i\int_{-\infty}^{x}(1-q(y)q^*(-y))dy \\
                                    iq^*(-x) \\
                                  \end{array}
                                \right)+\mathcal {O}(z^{-2}),\\
\mu_{+,2}(x;z)=e_{2}+\frac{1}{z}\left(
                                  \begin{array}{c}
                                    iq(x) \\
                                    i\int_{x}^{\infty}(1-q(y)q^*(-y))dy \\
                                  \end{array}
                                \right)+\mathcal {O}(z^{-2}),
\end{align}
and for $Imz<0$, as $z\rightarrow\infty$, the eigenfunctions $\mu_{\pm}$ follow that
\begin{align}\label{e21}
\mu_{+,1}(x;z)=e_{1}+\frac{1}{z}\left(
                                  \begin{array}{c}
                                    i\int_{x}^{\infty}(1-q(y)q^*(-y))dy \\
                                    iq^*(-x) \\
                                  \end{array}
                                \right)+\mathcal {O}(z^{-2}),\\
\mu_{-,2}(x;z)=e_{2}+\frac{1}{z}\left(
                                  \begin{array}{c}
                                    iq(x) \\
                                    i\int_{-\infty}^{x}(1-q(y)q^*(-y))dy \\
                                  \end{array}
                                \right)+\mathcal {O}(z^{-2}),
\end{align}
where $e_{1}=\left(
               \begin{array}{c}
                 1 \\
                 0 \\
               \end{array}
             \right)
$,  and  $e_{2}=\left(
               \begin{array}{c}
                 0 \\
                 1 \\
               \end{array}
             \right)
$.
\end{prop}
\begin{proof}
Using the similar method to the reference \cite{Prinari-PD}, the results of this proposition is proved easily.
\end{proof}
\begin{cor}
For the $q$ satisfied Proposition \ref{pp3}, as $z\rightarrow0$, the eigenfunctions $\mu_{\pm}$ follow that
\begin{align}\label{e22}
\mu_{\pm,1}(x;z)=\frac{i}{z}e_2+\mathcal {O}(1),~~
\mu_{\pm,2}(x;z)=-\frac{i}{z}e_1+\mathcal {O}(1).
\end{align}
\end{cor}
Then, based on the relationship \eqref{e14}, combined with the properties of $\mu_{\pm}$, the properties of the scattering matrix can be derived as follows.
\begin{cor}\label{S-Asy}
For  $\forall z\in\Sigma^{\circ}$, the determinant of scattering matrix is
\begin{align*}
          \det S(z)=1=s_{11}(z)s_{22}(z)-s_{21}(z)s_{12}(z).
        \end{align*}
As $z\to\infty$ na $z\to 0$, the asymptotic behavior of scattering matrix are
 \begin{align}\label{s-asy-infty-0}
  \begin{split}
            \lim_{z\to\infty}(s_{11}(z)-1)z=i\int_{\mathbb{R}}(q(x)q^*(-x)-1)\,dx;\\
            \lim_{z\to0}(s_{11}(z)+1)z^{-1}=i\int_{\mathbb{R}}(q(x)q^*(-x)-1)\,dx.
        \end{split}  \end{align}
\end{cor}

Furthermore,  reflection coefficient and the transmission coefficient are defined as
\begin{align}\label{r-def}
 r(z)=\frac{s_{21}(z)}{s_{11}(z)},\quad \check{r}(z)=\frac{s_{12}(z)}{s_{22}(z)}.
\end{align}
Then, according to Corollary \ref{S-Asy}, the asymptotic behavior of the reflection coefficient and transmission coefficient are derived as
\begin{align}\label{R-F-coefficient}
\begin{split}
r(z)\thicksim z^{-2},\quad z\to\infty,\quad \check{r}(z)\thicksim z^{-2},\quad z\to\infty,\\
r(z)\thicksim 0,\quad z\to0,\quad \check{r}(z)\thicksim 0,\quad z\to0.
\end{split}
\end{align}

Since the elements of the scattering matrix $S(z)$ shown in \eqref{e15}  have singularities at points $\pm1$,  we need to evaluate the boundedness of the reflection
coefficient $r(z)$ and transmission coefficient $\check{r}(z)$ at $z=\pm1$.  Actually, by carrying  out a simple calculation,  we obtain
\begin{align}\label{s-z-1}
\begin{split}
s_{11}(z)=\pm\frac{a_{\pm}}{z\mp1}+O(1),\quad s_{21}(z)=-\frac{a_{\pm}}{z\mp1}+O(1),\\
s_{22}(z)=\pm\frac{b_{\pm}}{z\mp1}+O(1),\quad s_{12}(z)=-\frac{b_{\pm}}{z\mp1}+O(1),
  \end{split}
\end{align}
where $a_{\pm}=\det\left(\psi_{+,1}(\pm1), \psi_{-,2}(\pm1)\right)$ and  $b_{\pm}\det\left(\psi_{-,1}(\pm1), \psi_{+,2}(\pm1)\right)$.
As a result, we can obtain the following results
\begin{align}\label{r-z-1}
  \lim_{z\to\pm1}r(z)=\mp1,\quad \lim_{z\to\pm1}\check r(z)=\mp1.
\end{align}

\section{A Riemann-Hilbert problem associated with initial value problem}

It is a fact that spectral singularities on real axis may  exist for initial data in any weighted Sobolev space $H^{j,k}$. Therefore, we select the general initial data which satisfies the following assumption to exclude the spectral singularities from the real axis. Then, many possible pathologies can be avoided.

\begin{assum}\label{assum}
In order to avoid the many possible pathologies, we make the following assumption, i.e.,
\begin{itemize}
  \item For $z\in\mathbb{R}$, no spectral singularities exist, i.e, $s_{11}(z)\neq0$ and $s_{22}(z)\neq0$;
  \item Suppose that $s_{11}(z)$ possesses $N_1$(finite) zero points and $s_{22}(z)$ possesses $N_2$(finite) zero points which are denoted as $\left\{(z_{j}, Re~z_{j}>0, Im~z_{j}<0)^{N_1}_{j=1}\right\}$ and $\left\{(\hat{z}_{j}, Re~\hat{z}_{j}>0, Im~\hat{z}_{j}>0)^{N_2}_{j=1}\right\}$.
  \item The discrete spectrum is simple, i.e., if $z_{0}$ is the zero of $s_{11}(z)$, then $s'_{11}(z_{0})\neq0$. At the same time, if $\hat{z}_{0}$ is the zero of $s_{22}(z)$, then $s'_{22}(\hat{z}_{0})\neq0$.
\end{itemize}
\end{assum}

Then, according to  the Assumption \ref{assum} and the symmetries of $S(z)$, the discrete spectrum set is $\mathcal{Z}\cup \hat{\mathcal{Z}}$, where $\mathcal{Z}=\{z_{j}, -z^*_{j},\frac{z^*_{j}}{|z_{j}|^2},-\frac{z_{j}}{|z_{j}|^2}\}_{j=1}^{N}$ and $\hat{\mathcal{Z}}=\{\hat{z}_{j}, -\hat z^*_{j},\frac{\hat z^*_{j}}{|\hat z_{j}|^2},-\frac{\hat z_{j}}{|\hat z_{j}|^2}\}_{j=1}^{N}$  (see Fig. \ref{fig02}).

\begin{figure}
\centerline{\begin{tikzpicture}[scale=0.9]
\path [fill=pink] (-4.5,0)--(-4.5,4.5) to (4.5,4.5) -- (4.5,0);
\draw[->][thick](3,0)--(5,0)node[right]{$Rez$};
\draw[->][thick](1.5,0)--(3,0);
\draw[->][thick](0,0)--(1.5,0);
\draw[->][thick](0,0)--(-1.5,0);
\draw[->][thick](-1.5,0)--(-3,0);
\draw[-][thick](-3,0)--(-5,0);
\draw[<->][thick](0,3)--(0,5)node[above]{$Imz$};
\draw[->][thick](0,3)--(0,1.5);
\draw[-][thick](0,1.5)--(0,0);
\draw[-][thick](0,-1.5)--(0,0);
\draw[->][thick](0,-4)--(0,-1.5);
\draw(0,0) [dashed][blue, line width=0.5] circle(2);
\draw[fill] (1.414,1.414) circle [radius=0.05];
\draw[fill] (1.414,1.414)node[above]{$\hat{z}_{j}$};
\draw[fill] (-1.414,1.414) circle [radius=0.05];
\draw[fill] (-1.414,1.414)node[above]{$-\hat{z}^*_{j}$};
\draw[fill] (1.414,-1.414) circle [radius=0.05];
\draw[fill] (-1.2,-1.8)node[left]{$-\frac{\hat{z}_{j}}{|\hat{z}_j|^2}=-\hat z_{j}$};
\draw[fill] (-1.414,-1.414) circle [radius=0.05];
\draw[fill] (1.2,-1.8)node[right]{$\frac{\hat z^*_{j}}{|\hat{z}_j|^2}=\hat{z}^*_{j}$};
\draw[fill] (3,4) circle [radius=0.05];
\draw[fill] (3,4)node[below]{$\hat z_i$};
\draw[fill] (-3,4) circle [radius=0.05];
\draw[fill] (-3,4)node[below]{$-\hat z^*_i$};
\draw[fill] (3,-4) circle [radius=0.05];
\draw[fill] (3,-4)node[above]{$z_i$};
\draw[fill] (-3,-4) circle [radius=0.05];
\draw[fill] (-3,-4)node[above]{$-z^*_i$};
\draw[fill] (0.6,-0.8) circle [radius=0.05];
\draw[fill] (0.6,-0.8)node[right]{$\frac{1}{\hat z_i}$};
\draw[fill] (-0.6,-0.8) circle [radius=0.05];
\draw[fill] (-0.6,-0.8)node[left]{$-\frac{1}{\hat z_i}$};
\draw[fill] (1,3.5)node[right]{$\mathbb C^+$};
\draw[fill] (1,-3.5)node[right]{$\mathbb C^-$};
\draw[fill] (-0.6,0.8) circle [radius=0.05];
\draw[fill] (-0.6,0.8)node[left]{$\frac{1}{-z^*_i}$};
\draw[fill] (0.6,0.8) circle [radius=0.05];
\draw[fill] (0.6,0.8)node[right]{$-\frac{1}{z_i}$};
\draw[fill] (1.732,1) circle [radius=0.05];
\draw[fill] (1.732,1)node[right]{$\frac{z^*_{j}}{|z_{j}|}=z^*_{j}$};
\draw[fill] (-1.732,1) circle [radius=0.05];
\draw[fill] (-1.732,1)node[left]{$\frac{-z_{j}}{|z_{j}|}=-z_{j}$};
\draw[fill] (1.732,-1) circle [radius=0.05];
\draw[fill] (1.732,-1)node[right]{$z_{j}$};
\draw[fill] (-1.732,-1) circle [radius=0.05];
\draw[fill] (-1.732,-1)node[left]{$-z^*_{j}$};
\end{tikzpicture}}
  \caption{\small (Color online) Analytical domains and distribution of the discrete spectrum $\mathcal{Z}$. $\mu_{-,1}$, $\mu_{+,2}$ and  $s_{22}$ are analytical in $\mathbb C^+$ (pink domain); $\mu_{+,1}$, $\mu_{-,2}$ and  $s_{11}$ are analytical in $\mathbb C^-$ (white domain); There are  discrete
spectrum on the green unit circle and  discrete spectrum outside the unit circle.}\label{fig02}
\end{figure}
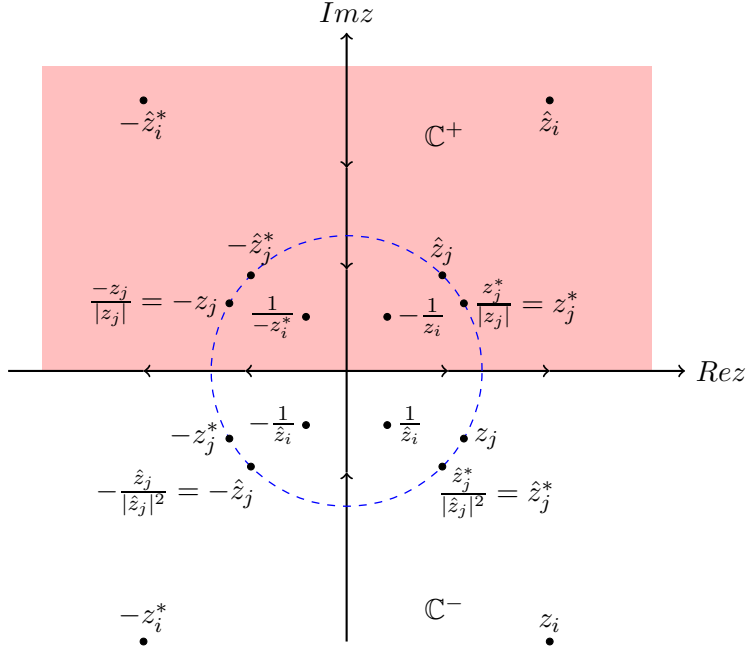

Define weighted Sobolev spaces
\begin{align}\label{W-H-Sobolev-spaces}
\begin{split}
&W^{k,p}(\mathbb{R})=\left\{f(x)\in L^{p}(\mathbb{R}):\partial^{j}f(x)\in L^{p}(\mathbb{R}), j=1,2,\ldots,k\right\},\\
&H^{k,2}(\mathbb{R})=\left\{f(x)\in L^{p}(\mathbb{R}):x^{2}\partial^{j}f(x)\in L^{p}(\mathbb{R}), j=1,2,\ldots,k\right\}.\\
&\mathcal{H}(\mathbb{R})=L^{2,4}(\mathbb{R})\cap H^{4}(\mathbb{R}).
\end{split}
\end{align}

Furthermore,  we have
\begin{prop}\label{prop-r-map}
Given the initial data $q_{0}(x)-1\in L^{1,2}(\mathbb{R})$, $q'(x)\in W^{1,1}(\mathbb{R})$ then $r(z)\in H^{1,1}(\mathbb{R})$ and  $\check{r}(z)\in H^{1,1}(\mathbb{R})$.
\end{prop}
\begin{proof}
According to the definition of reflection coefficient $r(z)$, we have
\begin{align*}
  ||r(z)||_{L^\infty(\mathbb{R})}\leqslant1.
\end{align*}
Then, combined with \eqref{R-F-coefficient}, we have
\begin{align}\label{r-infty}
  r(z)\to\pm\infty,\quad z\to\pm\infty.
\end{align}
Based on the analytical properties of scattering matrix $S(z)$, we know that the scattering coefficients $s_{11}(z)$ and $s_{22}(z)$ are continuous  for $z\in\mathbb{R}\setminus\{\pm1,0\}$. Therefore, the reflection coefficient $r(z)$ is continuous  for $z\in\mathbb{R}\setminus\{\pm1,0\}$. Furthermore, combined \eqref{r-infty} with \eqref{r-z-1}, we obtain that $r(z)$ is bounded in a neighborhood of  $\{\pm1,0\}$ and
\begin{align}
  r(z)\in L^1(\mathbb{R})\cap L^2(\mathbb{R}).
\end{align}

Next, we only need to prove $r'(z)\in L^2(\mathbb{R})$.

For a sufficiently small constant $\delta_0$, the map
\begin{align}
  q\mapsto\det \left(\psi_{+,1}(z)~~ \psi_{-,2}(z) \right)
\end{align}
and
\begin{align}
  q\mapsto\det \left(\psi_{-,1}(z)~~ \psi_{+,2}(z) \right)
\end{align}
are locally Lipschitz maps from
\begin{align}\label{r-1}
\{ q:q'\in W^{1,1}(\mathbb{R})~~ and ~~q\in L^{1,n+1}(\mathbb R)\}\mapsto W^{n,\infty}(\mathbb{R}\setminus(-\delta_0,\delta_0))
\end{align}
for $n\geqslant0$.
In fact, $q\mapsto\psi_{+,1}(z,0)$ is
a locally Lipschitz map with values in $W^{n,\infty}(\overline{\mathbb{C}}^+\setminus D(0,\delta_0),\mathbb{C}^2)$. It is also established for $q\mapsto\psi_{-,2}(z,0)$ and $q\mapsto\psi_{-,1}(z,0)$ by replacing $\overline{\mathbb{C}}^+$ with $\overline{\mathbb{C}}^-$. Combined this results with the asymptotic behavior of scattering coefficients $s_{11}$ and $s_{22}$, we have $q\mapsto\ r(z)$ is
a locally Lipschitz map from the domain in \eqref{r-1}  into
\begin{align}\label{r-1}
W^{n,\infty}(I_{\delta_0})\cap H^{n}(I_{\delta_0})
\end{align}
with
\begin{align*}
  I_{\delta_0}:=\mathbb{R} \setminus\left((-\delta_0,\delta_0)\cup(1-\delta_0,1+\delta_0)\cup (-1-\delta_0,-1+\delta_0) \right).
\end{align*}
Then, fix the constant $\delta_0>0$ sufficiently small such that the three intervals $dist(z,\pm1)\leq\delta_0$ and $|z|\leq\delta_0$ have empty intersection.
Let $|z-1|< \delta _0$. According to the definition of $a_+$, we have
\begin{align} \label{r-3}
	\begin{aligned} &
  	r (z)  =  \frac{s_{21}(z)}{s_{11}(z)}  =
	\frac{ \det [ \psi_{-,1}(z; x),\psi_{+,1} (z; x ) ] }{\det [ \psi_{+,1}(z; x),\psi_{-,2} (z; x) ] }
	= \frac{-a_++\int _1^z   F (s)ds} { a_++\int _1^z G (s)ds}
	\end{aligned}
\end{align}
where \begin{align*}
        F (s) =\partial _s\det [  \psi _{-,1}(s; x),\psi_{+,1} (s; x ) ],\\
        G (s)=\partial _s\det [ \psi_{+,1}(s; x),\psi_{-,2} (s; x) ].
      \end{align*}
Then, we can consider the following two cases:
\begin{itemize}
  \item For $a_+\neq 0$, from the above formula \eqref{r-3}, it is clear that $r'(z)$ is defined and bounded around $1$.
  \item For $a_+=0$, we obtain
  \begin{align}\label{r-4}
    r(z)=\frac{\int _1^z   F (s)ds} {\int _1^z G (s)ds}.
  \end{align}
  According to \eqref{e15}, we have
  \begin{align}\label{r-5}
    (z^2-1)s_{11}(z)=z^2\det(\psi_{+,1},\psi_{-,2}).
  \end{align}
  Since $a_+=0$, we know that $\det(\psi_{+,1},\psi_{-,2})\big|_{z=1}=0$. Then, differentiating \eqref{r-5} with respect to $z$ at $z=1$,    we obtain
  \begin{align*}
    2s_{11}(1)=\partial_{z}\det(\psi_{+,1},\psi_{-,2})\big|_{z=1}=G(1).
  \end{align*}
  Additionally,  the scattering coefficient $s_{11}(z)$ can be expressed as
  \begin{align*}
    s_{11}(z)=&\frac{1}{s_{22}(z)}+\frac{s_{12}(z)s_{21}(z)}{s_{22}(z)}\\
    =&\frac{1}{s_{22}(z)}+s_{21}(z)\check r(z).
  \end{align*}
  By using \eqref{r-z-1}, we obtain that as $z\to1$,
  \begin{align*}
    s_{11}\thicksim \frac{r}{s_{21}}+s_{21}\check r\thicksim -\frac{1}{s_{21}}-s_{21}.
  \end{align*}
  From the above formula, it is obvious that $s_{11}(z)\big|_{z\to1}=0$ only and only if $s_{21}=i$. While, according to $1=s_{11}(z)s_{22}(z)-s_{21}(z)s_{12}(z)$ and if $s_{21}=i$,  then $s_{12}=-i$,  we obtain
  \begin{align*}
    s_{11}(z)s_{22}(z)=2,
  \end{align*}
  which means $s_{11}(z)\neq0$.
  Summarize the above results,  we have $s_{11}(z)\neq0$ as $z\to 1$ which in turn give rise to $G(1)\neq0$.
  Therefor, the derivative $r'(z)$ is bounded near 1.
\end{itemize}
By using similar method,  we can obtain the same conclusion at $-1$.  At $z=0$ we can use the symmetries $r(z ^{-1})=\overline{r}(z)$ to conclude that $r$ vanishes at the origin. It follows that $r' \in L^2(\R)$.

Moreover, using similar method, we can obtain the same conclusion with respect to $\check r(z)$.
\end{proof}

Additionally, in order to prepare to prove the later Proposition, we further give the following result.

\begin{prop}\label{prop-r-estimate}
For the initial data $q_{0}(x)-1\in \mathcal{H}(\mathbb{R})$,  then the reflection coefficient and transmission coefficient  satisfy
\begin{align*}
  \|\log(1-r\check r)\|_{L^{p}(\mathbb R)}<\infty,\quad \forall p\geq1.
\end{align*}
\end{prop}
\begin{proof}
By using a similar method to \cite{Cuccagna-2016}, it is not hard to show the correctness of this Proposition by a direct calculation.
\end{proof}

Now, based on the above Assumption and Proposition, we can  construct a RH problem.

Firstly, we  introduce a sectionally  meromorphic matrix
\begin{align}\label{Matrix}
M(x,t;z)=\left\{\begin{aligned}
&M^{+}(x,t;z)=\left(\mu_{-,1}(x,t;z),\frac{\mu_{+,2}(x,t;z)}{s_{22}(z)}\right), \quad z\in \mathbb{C}^{+},\\
&M^{-}(x,t;z)=\left(\frac{\mu_{+,1}(x,t;z)}{s_{22}(z)},\mu_{-,2}(x,t;z)\right), \quad z\in \mathbb{C}^{-},
\end{aligned}\right.
\end{align}
where $M^{\pm}(x,t;z)=\lim\limits_{\varepsilon\rightarrow0^{+}}M(x,t;z\pm i\varepsilon),~\varepsilon\in\mathbb{R}$.
Then, the matrix function $M(x,t;z)$ satisfies the following matrix RHP.
\begin{RHP}\label{RH-1}
Find an analytic function $M(x,t;z)$ with the following properties:
\begin{itemize}
  \item $M(x,t;z)$ is meromorphic in $\mathbb{C}\setminus\mathbb{R}$;
  \item $M_{+}(x,t;z)=M_{-}(x,t;z)V(x,t;z)$,~~~$z\in\mathbb{R}$,
  where \begin{align}\label{J-Matrix-1}
V(x,t;z)=e^{-2it\theta(z)\hat{\sigma}_3}\left(\begin{array}{cc}
                   1 & r(z) \\
                   -\check{r}(z) & 1-r(z)\check{r}(z)
                 \end{array}\right).
\end{align}
  \item $M(x,t;z)=\mathbb{I}+O(z^{-1})$ as $z\rightarrow\infty$.
   \item $M(x,t;z)=-\frac{i}{z}\sigma_3Q_\pm+O(z)$ as $z\rightarrow0$.
  \item $M(x,t;z)$ possesses simple poles at each points in $\mathcal{Z}$ with:
      \begin{align}\label{res-tildeM}
      \begin{split}
\mathop{Res}_{z=z_{j}}M(x,t;z)&=\lim_{z\rightarrow z_{j}}M(x,t;z)\left(\begin{array}{cc}
                   0 &  \\
                   b_{j}e^{2it\theta(z_j)} & 0
                 \end{array}\right),\\
\mathop{Res}_{z=-z^*_{j}}M(x,t;z)&=\lim_{z\rightarrow -z^*_{j}}M(x,t;z)\left(\begin{array}{cc}
                   0 & 0 \\
                   b_{N_1+j}e^{2it\theta(-z^*_{j})} & 0
                 \end{array}\right),\\
  \mathop{Res}_{z=\frac{z_j^*}{|z_j|^2}}M(x,t;z)&=\lim_{z\rightarrow \frac{z_j^*}{|z_j|^2}}M(x,t;z)\left(\begin{array}{cc}
                   0 & d_{j}e^{-2it\theta(\frac{z_j^*}{|z_j|^2})} \\
                   0 & 0
                 \end{array}\right),\\
\mathop{Res}_{z=-\frac{z_j}{|z_j|^2}}M(x,t;z)&=\lim_{z\rightarrow -\frac{z_j}{|z_j|^2}}M(x,t;z)\left(\begin{array}{cc}
                   0 & d_{N_1+j}e^{-2it\theta(-\frac{z_j}{|z_j|^2})} \\
                   0 & 0
                 \end{array}\right),\\
\mathop{Res}_{z=\frac{1}{\hat{z}_{j}}}M(x,t;z)&=\lim_{z\rightarrow \frac{1}{\hat{z}_{j}}}M(x,t;z)\left(\begin{array}{cc}
                   0 &  \\
                   b_{2N_1+j}e^{2it\theta(\frac{1}{\hat{z}_{j}})} & 0
                 \end{array}\right),\\
\mathop{Res}_{z=-\frac{1}{\hat{z}^*_{j}}}M(x,t;z)&=\lim_{z\rightarrow -\frac{1}{\hat{z}^*_{j}}}M(x,t;z)\left(\begin{array}{cc}
                   0 & 0 \\
                   b_{2N_1+N_2+j}e^{2it\theta(-\frac{1}{\hat{z}^*_{j}})} & 0
                 \end{array}\right),\\
  \mathop{Res}_{z=\hat{z}_{j}}M(x,t;z)&=\lim_{z\rightarrow \hat{z}_{j}}M(x,t;z)\left(\begin{array}{cc}
                   0 & d_{2N_1+j}e^{-2it\theta(\hat{z}_{j})} \\
                   0 & 0
                 \end{array}\right),\\
\mathop{Res}_{z=-\hat{z}^*_{j}}M(x,t;z)&=\lim_{z\rightarrow -\hat{z}^*_{j}}M(x,t;z)\left(\begin{array}{cc}
                   0 & d_{2N_1+N_2+j}e^{-2it\theta(-\hat{z}^*_{j})} \\
                   0 & 0
                 \end{array}\right),
\end{split}
\end{align}
where
\begin{align*}
b_{j}=\frac{s_{21}(z_j)}{s'_{11}(z_j)},\quad b_{N_1+j}=\frac{s_{21}(-z^*_j)}{s'_{11}(-z^*_j)},\quad
b_{2N_1+j}=\frac{s_{21}(\frac{1}{\hat{z}_{j}})}{s'_{11}(\frac{1}{\hat{z}_{j}})},\quad
b_{2N_1+N_2+j}=\frac{s_{21}(-\frac{1}{\hat{z}^*_{j}})}{s'_{11}(-\frac{1}{\hat{z}^*_{j}})},\\
d_{j}=\frac{s_{12}(\frac{z_j^*}{|z_j|^2})}{s'_{22}(\frac{z_j^*}{|z_j|^2})}, \quad
d_{N_1+j}=\frac{s_{12}(-\frac{z_j}{|z_j|^2})}{s'_{22}(-\frac{z_j}{|z_j|^2})},\quad d_{2N_1+j}=\frac{s_{12}(\hat{z}_j)}{s'_{22}(\hat{z}_j)}, \quad
d_{2N_1+N_2+j}=\frac{s_{12}(-\hat{z}^*_j)}{s'_{22}(-\hat{z}^*_j)}.
\end{align*}
\end{itemize}
\end{RHP}
\begin{proof}
Based on the previous analysis, the first four results can be obtained through a direct calculation. For the fifth result, i.e., the residue conditions, by using \eqref{e14} and considering a fact that the discrete spectrum points are the zero points of the scattering data $s_{11}(z)$ and $s_{22}(z)$, we have
\begin{equation}
\mu_{+,1}(z_{j})=s_{21}(z_j)e^{2it\theta(z_j)}\mu_{-,2}(z_{j}).
\end{equation}
Then, we have
\begin{align*}
\mathop{Res}_{z=z_{j}}M(x,t;z)&=
\left(\mathop{Res}_{z=z_{j}}\left(\frac{\mu_{+,1}(z)}{s_{21}(z)} \right)~~0 \right)\\
&=\left(\frac{s_{21}(z_j)}{s'_{11}(z_j)}e^{2it\theta(z_j)}\mu_{-,2}(z_{j})~~0 \right).
\end{align*}
Define the notation $b_{j}:=\frac{s_{21}(z_j)}{s'_{11}(z_j)}$, the first formula in \eqref{res-tildeM} is obtained directly. The remaining formulae can be derived in a similar way.
\end{proof}

In order to reconstruct the solution $q(x,t)$ of the NNLS equation, it is necessary to study the asymptotic behavior of $M(x,t;z)$.

\begin{prop}\label{q-potential}
Given the initial data $q_{0}(x)-1\in  L^{1,2}(\mathbb{R})$, $q'(x)\in W^{1,1}(\mathbb{R})$, for $\pm Imz>0$,
\begin{align}\label{e20}
\lim_{z\to\infty}(zM(z)-\mathbb{I})&=\left(
                                  \begin{array}{cc}
                                    i\int^{+\infty}_{x}(1-q(y)q^*(-y))dy & -iq(x)\\
                                    -iq^*(-x) & i\int^{+\infty}_{x}(1-q(y)q^*(-y))dy\\
                                  \end{array}
                                \right),\\
\lim_{z\to0}(M(z)+\frac{\sigma_2}{z})&=\frac{1}{z}\left(
                                  \begin{array}{cc}
                                    -iq(x) & i\int_{x}^{+\infty}(1-q(y)q^*(-y))dy \\
                                    i\int_{x}^{+\infty}(1-q(y)q^*(-y))dy &  -iq^*(-x)\\
                                  \end{array}
                                \right)+\mathcal {O}(z^{-2}).
\end{align}
\end{prop}

\begin{proof}
From Proposition \eqref{pp3} and \eqref{s-asy-infty-0}, the asymptotic behavior at infinity is obtained  immediately. Based on the symmetry properties of eigenfunctions $\mu_{\pm}(z)$ and scattering matrix $S(z)$, we derive the following result
\begin{align}\label{sym-M-z}
  M(x,t;z)=-\frac{1}{z}M(x,t;\frac{1}{z})\sigma.
\end{align}
By using this symmetry property of $M(x,t;z)$, it is obvious to obtain  the asymptotic  behavior at the origin from the behavior at infinity.
\end{proof}

Then, according to the Proposition \ref{q-potential}, the potential $q(x,t)$ is obtained  by the reconstruction formula,
\begin{align}\label{q-reconstruction-potential}
  q(x,t)=i\lim_{z\to\infty}(zM(x,t,z))_{12}.
\end{align}

\section{Conjugation}\label{section-Conjugation}

In this section, we are going  to re-normalize the Riemann-Hilbert problem \ref{RH-1} such that it is well-behaved as $t\rightarrow\infty$ along any characteristic by introducing an auxiliary function  $T(z)$.

It is observed that the oscillation term appearing in jump matrix \eqref{J-Matrix-1}, is $e^{\pm 2it\theta(z)}$ which expressed as
\begin{align}\label{4-1}
e^{\pm2it\theta(z)},~~ \theta(z)=\frac{1}{2}\frac{x}{t}(z+\frac{1}{z})-\frac{1}{2}(z^2-\frac{1}{z^2}).
\end{align}
Furthermore, let $\xi=\frac{x}{2t}$, the real part of $2i\theta(z)$ can be written as
\begin{align}\label{4-21}
Re(2i\theta(z))=2\left(-\xi(Im z+\frac{Im z}{|z|^2})+Re z Im z+ \frac{Re z Im z}{|z|^4})\right).
\end{align}

In what follows, we main pay attention to the case $1<\xi<K$,  where  $K$ is an arbitrarily large real number.
Then, the decaying domains of the oscillation term can be derived and shown in Fig. \ref{fig-1}.\\

\begin{figure}
%

  \includegraphics[width=7.6cm,height=6.8cm,angle=0]{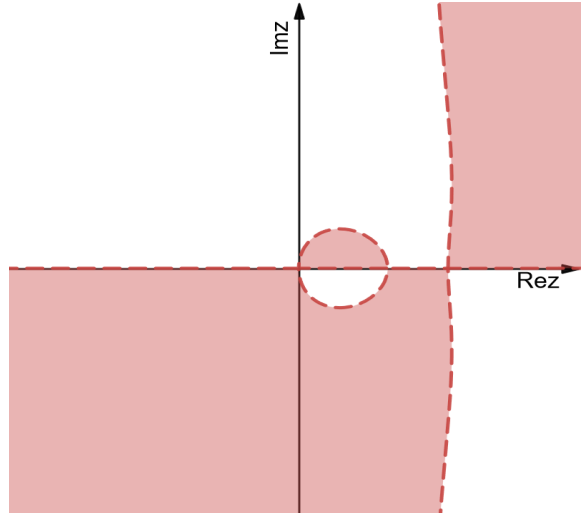}


  \caption{\small (Color online) The classification of sign $Re(2i\theta(z))$.  In the pink regions, $Re(2i\theta(z))>0$, which
implies that $e^{-2it\theta(z)}\to0$ as $t\to\infty$. While in the white regions, $Re(2i\theta(z))<0$, which implies $e^{2it\theta(z)}\to0$ as $t\to\infty$. The red curves $Re(2i\theta(z))=0$ are critical lines between decay and growth regions.}\label{fig-1}
\end{figure}

Then,  we are devoted to calculate the stationary phase points of the phase function $\theta(z)$.
\begin{align}\label{4-3}
\begin{split}
2\theta'(z)&=2\xi(1-z^{-2})-(2z+2z^{-3})\triangleq-2z^{-1}h(k),\\
2\theta''(z)&=2z^{-2}h(k)-2z^{-1}h'(k),
\end{split}
\end{align}
where $h(k)=4k^2-2\xi k-2$ and $k=\frac{1}{2}\left(z+\frac{1}{z}\right)$. Then, from \eqref{4-3}, the zero points of $2\theta'(z)$ can be derived as
\begin{align}\label{4-4}
\xi_1=\frac{1}{2}\left|k_0-\sqrt{k_0^2-4}\right|,\quad \xi_2=\frac{1}{2}\left|k_0+\sqrt{k_0^2-4}\right|,
\end{align}
where $k_0=\frac{1}{2}\left(|\xi|+\sqrt{\xi^2+8}\right)|$.
According to \eqref{4-3}, for $1<\xi<K$,  we have
\begin{align}\label{4-5}
0<\xi_1<\xi_2, \quad 2\theta''(\xi_1)>0, \quad 2\theta''(\xi_2)<0.
\end{align}

Next, based on the signal figure  of $Re(2i\theta(z))$ shown in Fig. \ref{fig-1},  different decomposition forms of the
jump matrix are chosen to guarantee  the oscillating term $e^{\pm2it\theta(z)}$  is decaying in all the
corresponding regions.

To make the following analyses easier to comprehend, 
we introduce some notations.
\begin{align*}
\nu_j=z_j,\quad \nu_{N_1+j}=-z_{j}^*,\quad \nu_{2N_1+j}=\frac{1}{\hat z_j}, \quad \nu_{2N_1+N_2+j}=-\frac{1}{\hat z^*_j}.
\end{align*}
\begin{align}\label{4-6}
\begin{aligned}
N=\{1,2,\ldots,2N_1+2N_2\},\quad
\bigtriangleup=\{j\in N |Re(i\theta(\nu_j))>0\},\\
\Lambda=\{j\in N \big||Re(i\theta(\nu_j))|<\delta_0\},\quad \bigtriangledown=\{j\in N |Re(i\theta(\nu_j))<0\},\\
N(\Lambda)=|\Lambda|,\quad \rho_0=\min_{j\in\mathbb N\setminus\Lambda}\{Im\theta(\nu_j)\}>\delta_0,\quad  I(\xi)=(0,\xi_1)\cup(\xi_2,+\infty),
\end{aligned}
\end{align}
where $\delta_0$ is an arbitrary small positive constant. To distinguish different type of zeros, we further make the following notations.
\begin{align}\label{4-7}
\begin{aligned}
\bigtriangleup_1&=\left\{i\in\{1,2,\ldots,N_1\} \big|Rei\theta(z_i)>0,Rei\theta(-z^*_i)>0\right\},\\ \bigtriangleup_2&=\left\{k\in\{1,2,\ldots,N_2\} \big|Rei\theta\left(\frac{1}{\hat z_k}\right)>0,Rei\theta\left(-\frac{1}{\hat z^*_k}\right)>0\right\},\\
\bigtriangledown_1&=\left\{i\in\{1,2,\ldots,N_1\} \big|Rei\theta(z_i)<0,Rei\theta(-z^*_i)<0\right\},\\ \bigtriangledown_2&=\left\{k\in\{1,2,\ldots,N_2\} \big|Rei\theta\left(\frac{1}{\hat z_k}\right)<0,Rei\theta\left(-\frac{1}{\hat z^*_k}\right)<0\right\},\\
\Lambda_1&=\left\{j_0\in\{1,2,\ldots,N_1\} \big||Rei\theta(z_{j_0}),Rei\theta(-z^*_{j_0})|<\delta_0\right\},\\ \Lambda_2&=\left\{j_0\in\{1,2,\ldots,N_2\} \big||Rei\theta\left(\frac{1}{\hat z_{j_0}}\right),Rei\theta\left(-\frac{1}{\hat z^*_{j_0}}\right)|<\delta_0\right\}.
\end{aligned}
\end{align}

Next, we going to re-normalize the Riemann-Hilbert problem \ref{RH-1}. Firstly, we introduce the following function
\begin{align*}
\delta(z)=\exp\left[-i\int_{I(\xi)}\nu(s)\left(\frac{1}{s-z}-\frac{1}{2s}\right)ds\right],
~~\nu(s)=\frac{1}{2\pi}\log(1-r(s)\check{r}(s)),
\end{align*}
and
\begin{align}\label{4-8}
T(z)=\prod_{k\in\bigtriangleup_1}\frac{(z+z^{*}_{k})(z-z_{k})}{(z^{*}_{k}z+1)(z_kz-1)}
\prod_{k\in\bigtriangleup_2}\frac{(z\check{z}^{*}_{k}+1)(z\check{z}_{k}-1)} {(z- \check{z}^{*}_{k})(z+\check{z}_k)} \delta(z).
\end{align}

Then, the function $T(z)$ possesses the following properties.

\begin{prop}\label{T-property} The function $T(z,\xi)$ satisfies that\\
($a$) $T$ is meromorphic in $C\setminus I(\xi)$. $T(z,\xi)$ possesses simple pole at $-z_{j}^{-1,*},-z_{j}^{-1}(j\in\bigtriangleup_1)$, $z_{k},-z_{jk}^{*}(k\in\bigtriangleup_2)$ and simple zero at $z_{j}, -z_j^*(j\in\bigtriangleup_1)$, $-\check{z}_{j}^{-1,*},-\check{z}_{j}^{-1}(k\in\bigtriangleup_2)$.\\
($b$) For $z\in I(\xi)$, $T_{+}(z)=T_{-}(z)\frac{1}{1-r(z)\check r(z)}$.\\
($c$) For $z\in C\setminus I(\xi)$, $T(z^{-1})=T^*(-z^*)=T(z)$.\\
($d$) For $z\to\infty$,
\begin{align*}
T(\infty):&=\lim_{z\to\infty}T(z)=\prod_{j\in\bigtriangleup_1}\frac{1}{|z_k|^2} \prod_{j\in\bigtriangleup_2}|\check{z}_k|^2 \exp\left(\frac{1}{4\pi i}\int_{I(\xi)}\frac{\log(1-r(s)\check r(s))}{s}\,ds\right);\\
|T(\infty)|&=\frac{|\check{z}_k|^2}{|z_{j}|^2},\quad j=1,2,\ldots,N_1,\quad k=1,2,\ldots,N_2.
\end{align*}
($e$) As $|z|\rightarrow \infty $,  the asymptotic expansion of $T(z)$ is expressed as
\begin{align}\label{4-9}
T(z,\xi)=&T(\infty) \left(1-\frac{1}{z}\mathop{\sum}\limits _{j\in\triangle_1}2i Im~z_{j}\left(1+ \frac{1}{|z_{k}|^2}\right)\right.\\
&\left.+\frac{1}{z}\mathop{\sum}\limits _{j\in\triangle_2}2i Im~\check{z}_{j}\left(1+ \frac{1}{|\check{z}_{k}|^2}\right)+ \frac{i}{z}\int_{I(\xi)}\nu(s)\,ds+O(z^{-2}) \right).
\end{align}
($f$) The ratio $\frac{s_{22}(z)}{T(z)}$ is holomorphic in $\mathbb{C}^+$ and there exists a constant $C(q_0)$ satisfying
\begin{align*}
 \big|\frac{s_{22}(z)}{T(z)}\big|\leqslant C(q_0).
\end{align*}
($g$)Local properties: For $k=1,2$,
\begin{align*}
\left|T(z)-T_k(\xi_k)(z-\xi_k)^{\epsilon_k\nu(\xi_k)i}\right|\lesssim \|\log(1-r(s)\check r(s)\|_{H^1}|z-\xi_k|^{\frac{1}{2}},
\end{align*}
where $T_k(\xi_k)=T(\infty)e^{i\beta_k(z;\xi_k)}$ with
\begin{align*}
\beta_k(z;\xi_k)=\epsilon_k\nu(\xi_k)\ln(z-\xi_k)+\int_{I(\xi)}\frac{\nu(s)}{s-z}\,ds, \quad \epsilon_k=\left\{\begin{aligned}
&1,\quad k=1,\\
&-1,\quad k=2.
\end{aligned}\right.
\end{align*}
\end{prop}

\begin{proof}
The above properties of $T(z)$ can be proved by a direct calculation, for details, see \cite{Cuccagna-2016,AIHP,Li-cgNLS}.
\end{proof}

In what follows, we use the function $T(z)$ to establish a transformation which reads
\begin{align}\label{Trans-1}
M^{(1)}(y,t;z)=M(y,t;z)T(z)^{\sigma_{3}}.
\end{align}

As a result, the RH problem \ref{RH-1} is transformed into the following RH problem for $M^{(1)}(y,t;z)$.
\begin{RHP}\label{RH-2}
Find an analytic function $M^{(1)}$ with the following properties:
\begin{itemize}
  \item $M^{(1)}$ is meromorphic on $\mathbb{C}\setminus I(\xi)$;
  \item $M^{(1)}(x,t,z)= \mathbb{I}+O(z^{-1})$ as $z\rightarrow \infty$;
  \item $M^{(1)}(x,t,z)= -\frac{i}{z}\sigma_3Q_{\pm}+O(z)$ as $z\rightarrow 0$;
  \item For $z\in \mathbb R$, the boundary values $M^{(1)}_{\pm}(z)$ satisfy the jump relationship $M^{(1)}_{+}(z)=M^{(1)}_{-}(z)V^{(1)}(z)$, where
      \begin{align}\label{4-10}
       V^{(1)}=\left\{\begin{aligned}
      \left(
        \begin{array}{cc}
      1 & 0 \\
      \check r(z)T(z)^{2}e^{2it\theta(z)} & 1 \\
        \end{array}
      \right)\left(
     \begin{array}{cc}
       1 & r(z)T(z)^{-2}e^{-2it\theta(z)} \\
       0 & 1 \\
      \end{array}
    \right),z\in  I(\xi),\\
   \left(
    \begin{array}{cc}
    1 & \frac{r(z)T_{-}^{-2}(z)}{1-r(z)\check r(z)}e^{-2it\theta(z)} \\
    0 & 1 \\
     \end{array}
   \right)\left(
    \begin{array}{cc}
    1 &  \\
    -\frac{\check r(z)T^{2}_{+}(z)}{1-r(z)\check r(z)}e^{2it\theta(z)} & 1 \\
   \end{array}
  \right),z\in \mathbb R\setminus I(\xi);
   \end{aligned}\right.
   \end{align}
   \item $M^{(1)}(z)$ has simple poles at each $\nu_{k}\in \mathcal{Z}\cup\check{\mathcal{Z}}$ and $\frac{1}{\nu_{k}}\in \mathcal{Z}\cup\check{\mathcal{Z}}$.\\
       For $k\in\bigtriangleup$,
\begin{align}\label{4-11}
\begin{split}
&\mathop{Res}\limits_{z=\nu_{k}}M^{(1)}(z)=\lim_{z\rightarrow \nu_{k}}M^{(1)}(z)\left(\begin{array}{cc}
    0 & b_{k}^{-1}\left(T'(\nu_{k})\right)^{-2}e^{-2it\theta(\nu_k)}\\
    0 & 0 \\
  \end{array}
\right),\\
&\mathop{Res}\limits_{z=\frac{1}{\nu_{k}}}M^{(1)}(z)=\lim_{z\rightarrow \frac{1}{\nu_{k}}}M^{(1)}(z)\left(\begin{array}{cc}
    0 & 0\\
    d_{k}^{-1}\left(\frac{1}{T}'(\frac{1}{\nu_{k}})\right)^{-2}e^{ 2it\theta(\frac{1}{\nu_{k}})} & 0 \\
  \end{array}
\right);
\end{split}
\end{align}
   For $k\in\bigtriangledown$,
   \begin{align}\label{4-12}
\begin{split}
&\mathop{Res}\limits_{z=\nu_{k}}M^{(1)}(z)=\lim_{z\rightarrow \nu_{k}}M^{(1)}(z)\left(\begin{array}{cc}
    0 & 0\\
    b_{k}\left(T(\nu_{k})\right)^{-2}e^{2it\theta(\nu_k)} & 0 \\
  \end{array}
\right),\\
&\mathop{Res}\limits_{z=\frac{1}{\nu_{k}}}M^{(1)}(z)=\lim_{z\rightarrow \frac{1}{\nu_{k}}}M^{(1)}(z)\left(\begin{array}{cc}
    0 & d_{k}T^{2}\left(\frac{1}{\nu_{k}}\right)e^{-2it\theta(\frac{1}{\nu_{k}})}\\
    0 & 0 \\
  \end{array}
\right).
\end{split}
\end{align}
\end{itemize}
\end{RHP}

\begin{proof}
Based on the above analysis, it is easy to verify the analyticity, jump  conditions, asymptotic behaviors and residue condition. For detail, see \cite{Li-cgNLS}.
\end{proof}

\section{Continuous extension to a mixed $\bar{\partial}$-RH problem}\label{Continuous-extension}

In this section, in order to re-normalize the  RH problem \ref{RH-2}, we extend the jump matrix off the real axis. It should be pointed out that a  continuous extension is sufficient.  With this extension, a new function can be defined to deform the  oscillation term along the real axis onto new contours. Along the  new contours, the deformed oscillation term are decaying. Therefore, we introduce an angle and some contours.

Firstly, fix a sufficient small angle $\theta_0$ such that the set $\{z\in\mathbb{C}: ~\cos\theta_0<\big|\frac{Rez}{z}\big|\}$ does not intersect any discrete spectrum points.

Define
\begin{align*}
\phi(\xi)=\min\{\theta_{0}\in(0,\frac{\pi}{4})\},
\end{align*}
and
\begin{align*}
\xi_0=0,\quad \xi_{0,1}=\frac{\xi_0+\xi_1}{2},\quad \xi_{1,2}=\frac{\xi_1+\xi_2}{2},\\
\ell_1\in\left(0,|\xi_{0,1}|\sec\phi(\xi)\right),\quad \ell_2\in\left(0,|\xi_{1,2}-\xi_1|\sec\phi(\xi)\right),\\
\tilde{\ell}_1\in\left(0,|\xi_{0,1}|\tan\phi(\xi)\right),\quad \tilde{\ell}_2\in\left(0,|\xi_{1,2}-\xi_1|\tan\phi(\xi)\right).
\end{align*}
We further define  $\Sigma_{kj}(k=0,1,2; j=1,2,3,4)$ and $\Sigma'=\bigcup_{j=1}^{4}\Sigma_j'$ where
\begin{align}\label{Sigma}
\begin{split}
\Sigma_{01}=\xi_0+e^{i(\pi-\phi(\xi))}\mathbb{R}_{+},\quad \Sigma_{04}=\overline{\Sigma_{01}},\\
\Sigma_{02}=\xi_0+e^{i(\pi-\phi(\xi))}\ell_{1},\quad \Sigma_{03}=\overline{\Sigma_{02}},\\
\Sigma_{11}=\xi_1+e^{i(\pi-\phi(\xi))}\ell_{1},\quad
\Sigma_{14}=\overline{\Sigma_{11}},\\
\Sigma_{12}=\xi_1+e^{i\phi(\xi)}\ell_{2},\quad
\Sigma_{13}=\overline{\Sigma_{12}},\\
\Sigma_{21}=\xi_2+e^{i\phi(\xi)}\mathbb{R}_{+},\quad
\Sigma_{24}=\overline{\Sigma_{21}},\\
\Sigma_{22}=\xi_2+e^{i(\pi-\phi(\xi))}\ell_2,\quad
\Sigma_{23}=\overline{\Sigma_{22}},\\
\Sigma_1'=\xi_{0,1}+e^{i\frac{\pi}{2}}\tilde{\ell}_1, \quad \Sigma_2'=\overline{\Sigma_1'},\\
\Sigma_3'=\xi_{1,2}+e^{i\frac{\pi}{2}}\tilde{\ell}_2, \quad \Sigma_4'=\overline{\Sigma_3'}.
\end{split}
\end{align}
Define $\Sigma^{(2)}=\bigcup_{j=1}^{4}\left(\bigcup_{k=1}^{2}\Sigma_{kj}\cup\Sigma_j'\right)$.
The regions divided by the boundary line $\Sigma_{kj}$ and $\Sigma'$ are denoted as $\Omega=\Omega_{kj}\cup\Omega_{up}\cup\Omega_{down}(k=0,1,2; j=1,2,3,4)$ and shown in Fig. \ref{fig-2}

\begin{figure}
\centerline{\begin{tikzpicture}[scale=0.65]
\path [fill=pink] (-3,0) -- (-1.5,0.45) to (0,0) -- (-1.5,-0.45);
\path [fill=pink] (4,0) -- (2,0.6) to (0,0) -- (2,-0.6);
\path [fill=pink] (4,0) -- (7,0.9) to (8,0.9) -- (8,0);
\path [fill=pink] (4,0) -- (7,-0.9) to (8,-0.9) -- (8,0);
\path [fill=pink] (-3,0) -- (-6,0.9) to (-8,0.9) -- (-8,0);
\path [fill=pink] (-3,0) -- (-6,-0.9) to (-8,-0.9) -- (-8,0);
\draw [dashed](-8,0)--(8,0);
\draw[->][thick](4,0)--(6,0.6);
\draw[-][thick](6,0.6)--(7,0.9);
\draw[-][thick](2,0.6)--(2,-0.6);
\draw[-][thick](-1.5,0.45)--(-1.5,-0.45);
\draw[->][thick](4,0)--(6,-0.6);
\draw[-][thick](6,-0.6)--(7,-0.9);
\draw[->][thick](4,0)--(3,0.3);
\draw[-][thick](3,0.3)--(2,0.6);
\draw[->][thick](4,0)--(3,-0.3);
\draw[-][thick](3,-0.3)--(2,-0.6);
\draw[->][thick](2,0.6)--(1,0.3);
\draw[->][thick](1,0.3)--(-0.75,-0.225);
\draw[->][thick](2,-0.6)--(1,-0.3);
\draw[->][thick](1,-0.3)--(-0.75,0.225);
\draw[-][thick](-0.75,0.225)--(-1.5,0.45);
\draw[->][thick](-1.5,0.45)--(-2.25,0.225);
\draw[-][thick](-2.25,0.225)--(-4.5,-0.45);
\draw[->][thick](-1.5,0.45)--(-2.25,0.225);
\draw[->][thick](-6,0.9)--(-4.5,0.45);
\draw[-][thick](-0.75,-0.225)--(-1.5,-0.45);
\draw[->][thick](-1.5,-0.45)--(-2.25,-0.225);
\draw[-][thick](-2.25,-0.225)--(-4.5,0.45);
\draw[->][thick](-6,-0.9)--(-4.5,-0.45);
\draw[fill] (0,0)node[below]{$\xi_1$} circle [radius=0.08];
\draw[fill] (4,0)node[below]{$\xi_2$} circle [radius=0.08];
\draw[fill] (-3,0)node[below]{$0$} circle [radius=0.08];
\draw[fill] (7,1.2)node[left]{$\Sigma_{21}$};
\draw[fill] (7,-1.2)node[left]{$\Sigma_{24}$};
\draw[fill] (-7,1.2)node[right]{$\Sigma_{02}$};
\draw[fill] (-7,-1.2)node[right]{$\Sigma_{03}$};
\draw[fill] (-1.5,0.9)node[left]{$\Sigma_{01}$};
\draw[fill] (-1.5,0.9)node[right]{$\Sigma_{11}$};
\draw[fill] (-1.5,-0.9)node[left]{$\Sigma_{04}$};
\draw[fill] (-1.5,-0.9)node[right]{$\Sigma_{14}$};
\draw[fill] (2,0.9)node[left]{$\Sigma_{12}$};
\draw[fill] (2,0.9)node[right]{$\Sigma_{22}$};
\draw[fill] (2,-0.9)node[left]{$\Sigma_{13}$};
\draw[fill] (2,-0.9)node[right]{$\Sigma_{23}$};
\draw[->][blue,dashed](-1.5,0.225)--(-0.2,0.45);
\draw[fill] (-0.4,0.45)node[right]{{\small$\Sigma_{1}'$}};
\draw[->][blue,dashed](-1.5,-0.225)--(0,-0.9);
\draw[fill] (0.2,-0.9)node[below]{\small{$\Sigma_{1}'$}};
\draw[->][blue,dashed](2,0.3)--(4,0.4);
\draw[fill] (3.7,0.4)node[right]{{\small$\Sigma_{2}'$}};
\draw[->][blue,dashed](2,-0.3)--(4,-0.9);
\draw[fill] (4.1,-0.9)node[below]{\small{$\Sigma_{2}'$}};
\draw(-3,0) [dashed][blue, line width=0.5] circle(4);
\end{tikzpicture}}

  \caption{\small (Color online) The boundary line $\Sigma_{kj}$ and $\Sigma'$.}\label{fig-2}
\end{figure}
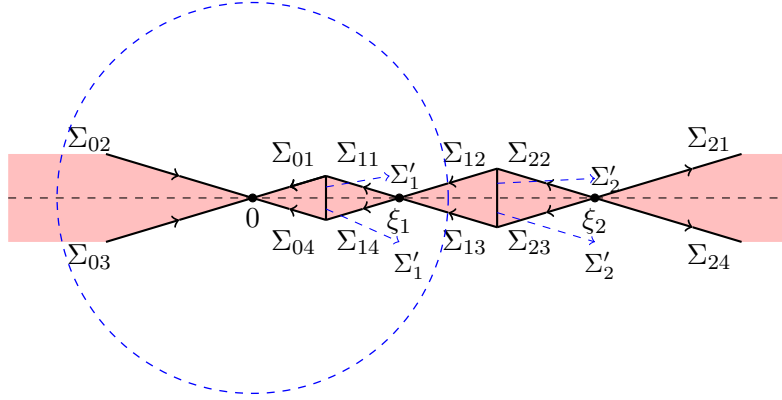

Additionally, for the case $\xi>1$, we define the intervals
\begin{align*}
  I_1=(-\infty,0),\quad I_2=(0,\xi_1),\quad I_3=(\xi_1,\xi_2),\quad I_4=(\xi_2,\infty).
\end{align*}

Next, in order to approach to the purpose of extending the jump matrix onto the new contours along which oscillation term are decaying, we will introduce transformation matrix $R^{(2)}$ which is defined in \eqref{R2-define}. Before this step, we give the following proposition to prepare for the next analysis.

\begin{prop}\label{Re-theta-estimation}
Let $K$ be a sufficiently large constant. Then, for $1<\xi<K$, the phase function $\theta(z)$ defined in \eqref{4-1} satisfies
\begin{align}\label{theta-estimate}
\begin{split}
Re\left(2i\theta(z) \right)>0,\quad z\in \Omega_{k1}\cup\Omega_{k3},\\
Re\left(2i\theta(z) \right)<0,\quad z\in \Omega_{k2}\cup\Omega_{k4},
\end{split}
\end{align}
where $k=0,1,2$.
\end{prop}

\begin{proof}
We only  consider the case $z\in \Omega _{01}$. The other cases can be proved similarly. For $1<\xi<K$ and $z\in \Omega_{01}$, we rewrite $z$ as
\begin{align*}
z=|z|e^{i\omega}=|z|(\cos\omega+i\sin\omega),
\end{align*}
where $\omega$ is  the argument of $z$ and $\omega\in (0,\phi(\xi))$. Then, based on $\theta(z)$ defined in \eqref{4-1}, we have
\begin{align*}
Re(2i\theta(z))&=2\left(-\xi\left(Imz+\frac{Imz}{|z|^2}\right) +RezImz+\frac{RezImz}{|z|^4}\right)\\
&=2\left(-\xi\left(|z|\sin\omega+\frac{|z|\sin\omega}{|z|^2}\right) +|z|^2\sin\omega \cos\omega+\frac{|z|^2\sin\omega \cos\omega}{|z|^4}\right)\\
&=\sin2\omega\left(\left(|z|+\frac{1}{|z|^2}\right)^2-2- \xi\left(|z|+\frac{1}{|z|^2}\right)\sec\omega\right).
\end{align*}
Then, taking $F(z)=z+\frac{1}{z}$, we have
\begin{align*}
Re(2i\theta(z))
&=\sin2\omega\left(F^2(|z|)-2-\xi\sec\omega F(|z|)\right)\\
&\triangleq\sin2\omega G(z).
\end{align*}
Observing that the function $G(z)=0$ possesses two zero points
\begin{align*}
F_1(|z|)=\frac{\xi\sec\omega-\sqrt{\xi^2\sec^2\omega}}{2},\\
F_2(|z|)=\frac{\xi\sec\omega+\sqrt{\xi^2\sec^2\omega}}{2}.
\end{align*}
Observing  that $F(|z|)\geqslant2$, we only need to check  whether $F_2(|z|)$ is greater than $2$. In other words, we only need to check the  correctness of the following formula
\begin{align*}
\xi\sec\omega+\sqrt{\xi^2\sec^2\omega}>4.
\end{align*}
For $\xi>1$ and $z\in\Omega_{01}$,  it is obvious that the above formula is positive which means  $G(z)$ is positive.
\end{proof}

We carry out the next step: extending the jump matrix onto the new contours along which oscillation term are decaying. Define
\begin{align}\label{Trans-2}
M^{(2)}=M^{(1)}R^{(2)},
\end{align}
where $R^{(2)}$ possesses some restrictions.
\begin{itemize}
  \item This step is to extend the jump matrix onto the new contours $\Sigma^{(2)}$. So after the matrix $R^{(2)}$ acts on matrix RH problem for $M^{(1)}$, the new matrix RH problem for $M^{(2)}$ must have no jump on the real axis.
  \item The norms of $R^{(2)}$ need to be controlled to ensure that the $\bar{\partial}$-contribution  has little influence on the long-time asymptotic solution of $q(x,t)$.
  \item The introduced transformation should have no influence on the residue condition.
\end{itemize}

Then, $R^{(2)}$ can be defined as
\begin{align}\label{R2-define}
R^{(2)}=\left\{\begin{aligned}
&\left(
  \begin{array}{cc}
    1 & (-1)^{m}R_{kj}e^{-2it\theta(z)}  \\
    0 & 1 \\
  \end{array}
\right), ~~&z\in\Omega_{kj},~~j=1,3,\\
&\left(
  \begin{array}{cc}
    1 & 0 \\
    (-1)^{m_{j}}R_{j}e^{-2it\theta} & 1 \\
  \end{array}
\right), ~~&z\in\Omega_{kj},~~j=2,4,\\
&\left(
  \begin{array}{cc}
    1 & 0 \\
    0 & 1 \\
  \end{array}
\right),~~ &z\in~~elsewhere,
\end{aligned}
\right.
\end{align}
where $m=\left\{\begin{aligned}0,~~~~~~j=1,4\\ 1,~~~~~~j=2,3\end{aligned}
\right.$ and $R_{kj}(z)$ are defined in the following proposition.

\begin{prop}\label{R-property}
Let $q_{0}-1\in L^{1,2}(\mathbb{R})$. Then,
there exist functions $\mathcal{R}:=R_{kj}: \bar{\Omega}_{kj} \rightarrow C, k=1,2;~j= 1, 2, 3, 4$ with boundary values such that
\begin{align*}
&R_{k1}(z)=\left\{\begin{aligned}&r(z)T^{-2}(z),  &z\in I_2\cup I_4,\\
&f_{k1}=r(\xi_k)T^{-2}(\xi_k)(z-\xi_k)^{2i\nu(\xi_k)\epsilon_k}, &z\in\Sigma_{k1},
\end{aligned}\right.\\
&R_{k2}(z)=\left\{\begin{aligned}&\frac{\check r(z)}{1-r(z)\check r(z)}T_{+}^{-2}(z) , &z\in I_1\cup I_3,\\
&f_{k2}=\frac{\check r(\xi_k)}{1-r(\xi_k)\check r(\xi_k)}T_{+}^{-2}(\xi_k) (z-\xi_k)^{-2i\nu(\xi_k)\epsilon_k}, &z\in \Sigma_{k2},
\end{aligned}\right.\\
&R_{k3}(z)=\left\{\begin{aligned}&\frac{r(z)}{1-r(z)\check r(z}T_{-}^{-2}(z), &z\in I_1\cup I_3,\\
&f_{k3}=\frac{r(\xi_k)}{1-r(\xi_k)\check r(\xi_k}T_{-}^{-2}(\xi_k) (z-\xi_k)^{2i\nu(\xi_k)\epsilon_k}, &z\in\Sigma_{k3},
\end{aligned}\right.\\
&R_{k4}(z)=\left\{\begin{aligned}&\check{r}(z)T^{2}(z), &z\in I_2\cup I_4,\\
&f_{k4}=\check{r}(\xi_k)T^{2}(\xi_k)(z-\xi_k)^{-2i\nu(\xi_k)\epsilon_k}, &z\in\Sigma_{k4}.
\end{aligned}\right.
\end{align*}
Additionally,  $R_{kj}$ possesses the following properties:
\begin{align}\label{R-estimate}
|\bar{\partial}R_{kj}(z)|\lesssim |r'(Rez)|+|z-\xi_k|^{-\frac{1}{2}}.
\end{align}
\end{prop}

\begin{proof}
The thinking method to prove the results
in Proposition \ref{R-property} is similar to that in \cite{AIHP,Li-cgNLS}. So we omit it.
\end{proof}

Then, applying the transformation \eqref{Trans-2} and the properties of $R^{(2)}$, we derive the following mixed $\bar{\partial}$-RH problem for $M^{(2)}$.

\begin{RHP}\label{RH-3}
Find a matrix value function $M^{(2)}$, admitting
\begin{itemize}
 \item $M^{(2)}(x,t,z)$ is continuous in $\mathbb{C}\setminus(\Sigma^{(2)}\cup\mathcal{Z}\cup\check{\mathcal{Z}})$.
 \item $M^{(2)}(x,t,z)= \mathbb{I}+O(z^{-1})$ as $z\rightarrow \infty$.
  \item $M^{(2)}(x,t,z)= -\frac{i}{z}\sigma_3Q_{\pm}+O(z)$ as $z\rightarrow 0$.
 \item $M_+^{(2)}(x,t,z)=M_{-}^{(2)}(x,t,z)V^{(2)}(x,t,z),$ ~~ $z\in\Sigma^{(2)}$, where the jump matrix $V^{(2)}(x,t,z)$ satisfies
 \begin{align}\label{V2-jump}
V^{(2)}=\left\{\begin{aligned}
&\left(
  \begin{array}{cc}
    1 & -f_{kj}e^{-2it\theta(z)}  \\
    0 & 1 \\
  \end{array}
\right), ~~&z\in\Sigma_{kj},~~j=1,3,\\
&\left(
  \begin{array}{cc}
    1 & 0 \\
    f_{kj}e^{2it\theta(z)} & 1 \\
  \end{array}
\right), ~~&z\in\Sigma_{kj},~~j=2,4,\\
&\left(
  \begin{array}{cc}
    1 & (f_{(k-1)j}-f_{kj})e^{-2it\theta(z)}  \\
    0 & 1 \\
  \end{array}
\right), ~~&z\in\Sigma_{j}',~~j=1,4,\\
&\left(
  \begin{array}{cc}
    1 & 0 \\
    -(f_{(k-1)j}-f_{kj})e^{2it\theta(z)} & 1 \\
  \end{array}
\right), ~~&z\in\Sigma_{j}',~~j=2,3,
\end{aligned}
\right.
\end{align}
where $k=1,2$.
\item For $\mathbb{C}\setminus(\Sigma^{(2)}\cup\mathcal{Z}\cup\check{\mathcal{Z}})$, $\bar{\partial}M^{(2)}=M^{(2)}\bar{\partial}\mathcal{R}^{(2)}(z),$ where
   \begin{align}\label{dbar-R2}
\bar{\partial}R^{(2)}=\left\{\begin{aligned}
&\left(
  \begin{array}{cc}
    1 & (-1)^{m}\bar{\partial}R_{kj}e^{-2it\theta(z)}  \\
    0 & 1 \\
  \end{array}
\right), ~~&z\in\Omega_{kj},~~j=1,3,\\
&\left(
  \begin{array}{cc}
    1 & 0 \\
    (-1)^{m}\bar{\partial}R_{kj}e^{2it\theta(z)} & 1 \\
  \end{array}
\right), ~~&z\in\Omega_{kj},~~j=2,4,\\
&\left(
  \begin{array}{cc}
    0 & 0 \\
    0 & 0 \\
  \end{array}
\right),~~ &z\in ~~elsewhere.
\end{aligned}
\right.
\end{align}
where $m=\left\{\begin{aligned}0,~~~~~~j=1,4\\ 1,~~~~~~j=2,3\end{aligned}
\right.$ and $k=0,1,2$.
  \item  $M^{(2)}$ admits the residue conditions at poles  $\nu_{k}\in \mathcal{Z}\cup\check{\mathcal{Z}}$ and $\frac{1}{\nu_{k}}\in \mathcal{Z}\cup\check{\mathcal{Z}}$. i.e.,\\
     For $k\in\bigtriangleup$,
\begin{align}\label{M2-Res-1}
\begin{split}
&\mathop{Res}\limits_{z=\nu_{k}}M^{(2)}(z)=\lim_{z\rightarrow \nu_{k}}M^{(2)}(z)\left(\begin{array}{cc}
    0 & b_{k}^{-1}\left(T'(\nu_{k})\right)^{-2}e^{-2it\theta(\nu_k)}\\
    0 & 0 \\
  \end{array}
\right),\\
&\mathop{Res}\limits_{z=\frac{1}{\nu_{k}}}M^{(2)}(z)=\lim_{z\rightarrow \frac{1}{\nu_{k}}}M^{(2)}(z)\left(\begin{array}{cc}
    0 & 0\\
    d_{k}^{-1}\left(\frac{1}{T}'(\frac{1}{\nu_{k}})\right)^{-2}e^{ 2it\theta(\frac{1}{\nu_{k}})} & 0 \\
  \end{array}
\right);
\end{split}
\end{align}
   For $k\in\bigtriangledown$,
   \begin{align}\label{M2-Res-2}
\begin{split}
&\mathop{Res}\limits_{z=\nu_{k}}M^{(2)}(z)=\lim_{z\rightarrow \nu_{k}}M^{(2)}(z)\left(\begin{array}{cc}
    0 & 0\\
    b_{k}\left(T(\nu_{k})\right)^{-2}e^{2it\theta(\nu_k)} & 0 \\
  \end{array}
\right),\\
&\mathop{Res}\limits_{z=\frac{1}{\nu_{k}}}M^{(2)}(z)=\lim_{z\rightarrow \frac{1}{\nu_{k}}}M^{(2)}(z)\left(\begin{array}{cc}
    0 & d_{k}T^{2}\left(\frac{1}{\nu_{k}}\right)e^{-2it\theta(\frac{1}{\nu_{k}})}\\
    0 & 0 \\
  \end{array}
\right).
\end{split}
\end{align}
\end{itemize}
\end{RHP}

Then, from the RH problem \ref{RH-3}, the solution of the NNLS equation \eqref{e2} is expressed as
\begin{align}\label{q-sol-M1}
  q(x,t)=i\lim_{z\to\infty}(zM^{(2)}(x,t,z))_{12}.
\end{align}

In what follows,  to solve the RH problem \ref{RH-3}, we first pay attention to the pure RHP by ignoring the $\bar{\partial}$ component of $\bar{\partial}$-RH problem \ref{RH-3}. The next step is to deal with   the remaining $\bar{\partial}$ problem. Therefore, in the next part, we are going to decompose the mixed $\bar{\partial}$-RH problem \ref{RH-3} and investigate them respectively.

\section{Decomposition of the mixed $\bar{\partial}$-RH problem}\label{Decomposition-mixed-RH-problem}

In this section, we are going to decompose the mixed $\bar{\partial}$-RH problem into two parts , including a model RH problem with $\bar{\partial}R^{(2)}=0$ and a pure $\bar{\partial}$-RH problem with $\bar{\partial}R^{(2)}\neq0$. We denote $M^{(2)}_{RHP}(z)=M^{(2)}_{RHP}(x,t;z)$ as the solution of the model RH problem, and construct a RH problem for $M^{(2)}_{RHP}$ first.

\begin{RHP}\label{RH-rhp}
Find a matrix value function $M^{(2)}_{RHP}(z)$, admitting
\begin{itemize}
 \item $M^{(2)}_{RHP}(z)$ is analytic in $\mathbb{C}\backslash(\Sigma^{(2)}\cup\mathcal{Z}\cup\check{\mathcal{Z}})$;
 \item $M^{(2)}_{RHP}(z)= \mathbb{I}+O(z^{-1})$ as $z\rightarrow \infty$;
 \item $M^{(2)}_{RHP}(z)= -\frac{i}{z}\sigma_3Q_{\pm}+O(z)$ as $z\rightarrow 0$;
 \item $M^{(2)}_{RHP,+}(z)=M^{(2)}_{RHP,-}(z)V^{(2)}(z),$ \quad $z\in\Sigma^{(2)}$, where $V^{(2)}(x,t,z)$ is the same with the jump matrix appeared in RHP \ref{RH-3};
 \item $M^{(2)}_{RHP}(z)$ possesses the same residue condition with $M^{(2)}$.
 \end{itemize}
\end{RHP}

We will give the proof of the existence of the solution of $M^{(2)}_{RHP}$ in section \ref{section-pure RH problem}. Then, using $M^{(2)}_{RHP}$, the RH problem \ref{RH-3} can be reduced to a pure $\bar{\partial}$-RH problem. By constructing a transformation
\begin{align}\label{delate-pure-RHP}
M^{(3)}(z)=M^{(2)}(z)M^{(2)}_{RHP}(z)^{-1},
\end{align}
we obtain the following pure $\bar{\partial}$-RH problem.
\begin{RHP}\label{RH-4}
Find a matrix value function $M^{(3)}$, satisfying
\begin{itemize}
 \item $M^{(3)}$ is continuous with sectionally continuous first partial derivatives in $\mathbb{C}\backslash(\Sigma^{(2)}\cup\mathcal{Z}\cup\check{\mathcal{Z}})$;
 \item $M^{(3)}(z)= \mathbb{I}+O(z^{-1})$ as $z\rightarrow \infty$.
 \item $M^{(3)}(z)= -\frac{i}{z}\sigma_3Q_{\pm}+O(z)$ as $z\rightarrow 0$.
 \item For $z\in \mathbb{C}$, we obtain $\bar{\partial}M^{(3)}(z)=M^{(3)}(z)W^{(3)}(z)$,
       where
       \begin{align}\label{5-4}
       W^{(3)}=M_{RHP}^{(2)}(z)\bar{\partial}R^{(2)}M_{RHP}^{(2)}(z)^{-1}.
       \end{align}
 \end{itemize}
\end{RHP}
\begin{proof}
According to the properties of the $M^{(2)}_{RHP}$ and $M^{(2)}$ for RHP \ref{RH-rhp} and RHP \ref{RH-3}, the analytic and asymptotic properties of $M^{(3)}$ can be derived easily. Noting the fact that $M^{(2)}_{RHP}$ possesses the same jump matrix with $M^{(2)}$, we obtain that
\begin{align*}
M^{(3)}_{-}(z)^{-1}M^{(3)}_{+}(z)&=M^{(2)}_{RHP,-}(z)M^{(2)}_{-}(z)^{-1}M^{(2)}_{+}(z)M^{(2)}_{RHP,+}(z)^{-1}\\
&=M^{(2)}_{RHP,-}(z)V^{2}(z)(M^{(2)}_{RHP,-}(z)V^{2}(z))^{-1}=\mathbb{I},
\end{align*}
which implies that $M^{(3)}$ has no jump. Also, it is easy to prove that there exists no pole in $M^{(3)}$ by a simple analysis. For details, see \cite{AIHP,Li-cgNLS}.
\end{proof}

Then, in what follows, we respectively study the RH problem \ref{RH-rhp} for $M^{(2)}_{RHP}$ and the pure $\bar{\partial}$-RH problem \ref{RH-4} for $M^{(3)}$.

\section{Solving the pure RH problem}\label{section-pure RH problem}

In this section, we are going to construct the solution $M^{(2)}_{RHP}$ of RH problem \ref{RH-rhp}.

Firstly, we define
\begin{align*}
\rho=\frac{1}{2}\min\left\{\min_{j\in\mathcal N}\{Im\nu_j\},\min_{j\in\mathcal N\setminus\Lambda,Im(\theta(z))=0}\{|z-\nu_j|\},\min_{n=0,\pm1} \{|\xi_1-n|,|\xi_2-n|\},\min_{i,j\in\mathcal N,i\neq j}\{|\nu_i-\nu_j|\}\right\}.
\end{align*}
Then, we further define
\begin{align*}
\mathcal{U}_{\xi}=\mathcal{U}_{\xi_1}\cup\mathcal{U}_{\xi_2},\quad \mathcal{U}_{\xi_k}=\left\{z:|z-\xi_k|<\rho, k=1,2\right\}.
\end{align*}

Next, we   decompose $M^{(2)}_{RHP}$ into two parts i.e.,
\begin{align}\label{Mrhp}
M^{(2)}_{RHP}(z)=\left\{\begin{aligned}
&E(z)M^{out}(z), &&z\in\mathbb{C}\setminus\mathcal{U}_{\xi},\\
&E(z)M^{out}(z)M^{(loc)}(z), &&z\in \mathcal{U}_{\xi}.
\end{aligned} \right.
\end{align}
According to the above decomposition and the definition of $\rho$,  we know that $M^{(loc)}(z)$ possesses no poles in $\mathcal{U}_{\xi}$.
Moreover, $M^{out}$ solves a model RH problem,  $M^{(\xi)}$ can be solved by matching  a known parabolic cylinder model in $\mathcal{U}_{\xi}$, and $E(z)$ is an error function which is a solution of a small-norm Riemann-Hilbert problem.

In addition, to prepare for the next analysis, we study the estimates of the jump matrix $V^{(2)}$  in advance.
\begin{prop}\label{V2-esti}
As $t\to\infty$, there exist positive constants $c_1,c_2$ such that
\begin{align}
||V^{(2)}-\mathbb{I}||_{L^{\infty}(\Sigma^{(2)})}
=\left\{\begin{aligned}&O(e^{-c_1t}),\quad &z\in\Sigma^{(2)}\setminus\mathcal U_{\xi},\\
&O(e^{-c_2t}),\quad &z\in\Sigma',\\
&O(|z-\xi_k|t^{-\frac{1}{2}}),\quad &z\in\Sigma^{(2)}\cap\mathcal U_{\xi}.
\end{aligned} \right.
\label{V2-Est-2}
\end{align}
\end{prop}
\begin{proof}
The thinking method to prove the results in Proposition \ref{V2-esti} is similar to that in \cite{Fan-dnls-nonzero}. So we omit it.
\end{proof}

The Proposition \ref{V2-esti} shows  that if we omit the jump condition of  $M^{(2)}_{RHP}(z)$, there only exists exponentially small error with respect to $t$ outside the $\mathcal{U}_{\xi}$. Moreover, since
$V^{(2)}\rightarrow I$ as $z\rightarrow 0$,  it is not necessary to study  the  neighborhood of $z=0$ alone.

\subsection{Outer model RH problem: $M^{out}$}

In this subsection, we establish a model RH problem for $M^{out}$ and give the proof that the solution of $M^{out}$ can be approximated by a finite sum of soliton solutions.

The model RH problem for $M^{out}$ satisfies the RH problem as follows:
\begin{RHP}\label{RH-5}
Find a matrix value function $M^{out}(z)$, admitting
\begin{itemize}
  \item $M^{out}(z)$ is analytical in $\mathbb{C}\setminus(\Sigma^{(2)}\cup\mathcal{Z}\cup\check{\mathcal{Z}})$;
  \item $M^{out}(z)= \mathbb{I}+O(z^{-1})$ as $z\rightarrow \infty$.
  \item $M^{out}(z)= -\frac{i}{z}\sigma_3Q_{\pm}+O(z)$ as $z\rightarrow 0$.
  \item $M^{out}(z)$ has simple poles at  each $\nu_{k}\in \mathcal{Z}\cup\check{\mathcal{Z}}$ and $\frac{1}{\nu_{k}}\in \mathcal{Z}\cup\check{\mathcal{Z}}$, admitting the same residue condition in RH problem \ref{RH-3} by replacing $M^{(2)}(z)$ with $M^{out}(z)$.
\end{itemize}
\end{RHP}

\begin{prop}
For given scattering data $\sigma_d=\left\{\nu_k,\frac{1}{\nu_k},r(z),\check r(z)\right\}_{k=1,2,\ldots,2N_1+2N_2}$, if  $M^{out}(z)$ is the solution of RH problem \ref{RH-4}, then  $M^{out}(z)$ exists and is unique. Moreover, the solution $M^{out}(z)$ is equivalent to  the solution of RH problem for $M(z)$ under the reflectionless condition with modified connection coefficients $\tilde b_k$ and $\tilde d_k$ which are defined in \eqref{c-c}.
\end{prop}
\begin{proof}
Firstly, to transform  $M^{out}(z)$ to the soliton-solution of RHP \ref{RH-1}, we need to restore the jumps.
Reversing  transform \eqref{Trans-1} and \eqref{Trans-2}, we have
\begin{align}\label{N-revers}
N(z,D)=\left(\prod_{j\in\bigtriangleup_1,k\in\bigtriangleup_2} \frac{|\check{z}_k|^2}{|z_{j}|^2}\right)^{\sigma_3}M^{out}(z) \left(\prod_{k\in\bigtriangleup_1}\frac{(z+z^{*}_{k})(z-z_{k})}{(z^{*}_{k}z+1)(z_kz-1)}
\prod_{k\in\bigtriangleup_2}\frac{(z\check{z}^{*}_{k}+1)(z\check{z}_{k}-1)} {(z- \check{z}^{*}_{k})(z+\check{z}_k)}\right)^{-\sigma_3},
\end{align}
where $D=\left\{r(z),\check r(z),\{\nu_k,\frac{1}{\nu_k},b_k\delta^{-2},d_k\delta^{2}\}_{k\in\mathcal N}\right\}$. Next, we check that $N(z,D)$ satisfies RH problem \ref{RH-1}. It is obvious that $\tilde{M}(z)$ preserves the normalization conditions at origin and infinity. Furthermore, according to Proposition \ref{V2-esti} and comparing \eqref{N-revers}  with \eqref{V2-jump}, we know that $N(z,D)$ has no jump in the region  $\mathbb C\setminus\mathcal U_{\xi}$.
Additionally, $N(z,D)$ has the same residue condition with \eqref{res-tildeM} by replacing $b_{k}$ and $d_k$ with $\tilde b_k$ and $\tilde d_k$ where
\begin{align}\label{c-c}
\begin{split}
\tilde{b}_k&=b_ke^{\frac{1}{2\pi i}\int_{I(\xi)}\log(1-r(s)\check r(z))\left(\frac{1}{s-\zeta_j}-\frac{1}{2s}\right)\,ds},\\
\tilde{d}_k&=d_ke^{-\frac{1}{2\pi i}\int_{I(\xi)}\log(1-r(s)\check r(z))\left(\frac{1}{s-\gamma_j}-\frac{1}{2s}\right)\,ds}.
\end{split}
\end{align}
Therefore, with scattering data $\{r(z)\equiv0, \check r(z)\equiv0, \{\nu_k,\frac{1}{\nu_k}, b_k, d_k\}_{k\in\mathcal N}\}$, RH problem \ref{RH-1} has a solution $N(z,D)$. According to the above conclusions, we obtain that $M^{out}(z)$ is
the solution of RH problem \ref{RH-1} corresponding to a $N$-soliton, reflectionless, potential $\tilde{q}(x,t)$ which generates the same discrete spectrum $\mathcal{Z}\cup\check{\mathcal{Z}}$ as our initial data, but whose connection coefficients \eqref{c-c} are perturbations of those for the original initial data by an amount related to the reflection coefficient of the initial data.
\end{proof}

Although $M^{out}(z)$ exists and is unique, as $t\to\infty$, not all discrete spectral points contribute to the solution $M^{out}(z)$. Therefore, we  trade the poles for jumps on small contours encircling each pole.
Define
\begin{align}\label{Trans-H}
H(z)=\left\{\begin{aligned}
      &\left(
        \begin{array}{cc}
      1 & 0 \\
      -\frac{b_{j}e^{2it\theta(\nu_{j})}}{z-\nu_{j}} & 1 \\
        \end{array}
      \right), ~~&|z-\nu_{j}|<\rho,\quad j\in\bigtriangledown\setminus\Lambda,\\
   &\left(
        \begin{array}{cc}
      1 & -\frac{z-\nu_{j}}{b_{j}e^{2it\theta(\nu_{j})}} \\
      0 & 1 \\
        \end{array}
      \right), ~~&|z-\nu_{j}|<\rho, \quad j\in\bigtriangleup\setminus\Lambda,\\
    &\left(
        \begin{array}{cc}
      1 & -\frac{d_{j}e^{-2it\theta(\frac{1}{\nu_{j}})}}{z-\frac{1}{\nu_{j}}} \\
      0 & 1 \\
        \end{array}
      \right), ~~&|z-\frac{1}{\nu_{j}}|<\rho,\quad j\in\bigtriangledown\setminus\Lambda,\\
      &\left(
        \begin{array}{cc}
      1 &  0 \\
      -\frac{z-\frac{1}{\nu_{j}}}{d_{j}e^{-2it\theta(\frac{1}{\nu_{j}})}} & 1 \\
        \end{array}
      \right), ~~&|z-\frac{1}{\nu_{j}}|<\rho, \quad j\in\bigtriangleup\setminus\Lambda.
   \end{aligned}\right.
\end{align}
Then, applying $H(z)$, we introduce the transformation
\begin{align}\label{Trans-H-M}
\tilde M^{(2)}=\left\{\begin{aligned}
&M^{(2)}H(z),\\
&M^{(2)}, \quad elsewhere.
\end{aligned}\right.
\end{align}
Then, we obtain the following RH problem from RHP \ref{RH-3}.
\begin{RHP}\label{RH-tilde-M2}
Find a matrix value function $\tilde M^{(2)}$, admitting
\begin{itemize}
 \item $\tilde M^{(2)}(x,t,z)$ is continuous in $\mathbb{C}\setminus(\Sigma^{(2)}\cup\{|z-\nu_j|=\rho,|z-\frac{1}{\nu_j}|=\rho\})$.
 \item $\tilde M^{(2)}(x,t,z)= \mathbb{I}+O(z^{-1})$ as $z\rightarrow \infty$.
 \item $\tilde M^{(2)}(x,t,z)= -\frac{i}{z}\sigma_3Q_{\pm}+O(z)$ as $z\rightarrow 0$.
 \item $\tilde M_+^{(2)}(x,t,z)=\tilde M_{-}^{(2)}(x,t,z)\tilde V^{(2)}(x,t,z),$ ~~ $z\in\Sigma^{(2)}\cup\{|z-\nu_j|=\rho,|z-\frac{1}{\nu_j}|=\rho\}$, where the jump matrix $\tilde V^{(2)}(x,t,z)$ satisfies
 \begin{align}\label{V2-tilde-jump}
V^{(2)}=\left\{\begin{aligned}
&\left(
  \begin{array}{cc}
    1 & -f_{kj}e^{-2it\theta(z)}  \\
    0 & 1 \\
  \end{array}
\right), ~~&z\in\Sigma_{kj},~~j=1,3,\\
&\left(
  \begin{array}{cc}
    1 & 0 \\
    f_{kj}e^{2it\theta(z)} & 1 \\
  \end{array}
\right), ~~&z\in\Sigma_{kj},~~j=2,4,\\
&\left(
  \begin{array}{cc}
    1 & (f_{(k-1)j}-f_{kj})e^{-2it\theta(z)}  \\
    0 & 1 \\
  \end{array}
\right), ~~&z\in\Sigma_{j}',~~j=1,4,\\
&\left(
  \begin{array}{cc}
    1 & 0 \\
    -(f_{(k-1)j}-f_{kj})e^{2it\theta(z)} & 1 \\
  \end{array}
\right), ~~&z\in\Sigma_{j}',~~j=2,3,\\
 &\left(
        \begin{array}{cc}
      1 & 0 \\
      -\frac{b_{j}e^{2it\theta(\nu_{j})}}{z-\nu_{j}} & 1 \\
        \end{array}
      \right), ~~&|z-\nu_{j}|=\rho,\quad j\in\bigtriangledown\setminus\Lambda,\\
   &\left(
        \begin{array}{cc}
      1 & -\frac{z-\nu_{j}}{b_{j}e^{2it\theta(\nu_{j})}} \\
      0 & 1 \\
        \end{array}
      \right), ~~&|z-\nu_{j}|=\rho, \quad j\in\bigtriangleup\setminus\Lambda,\\
    &\left(
        \begin{array}{cc}
      1 & -\frac{d_{j}e^{-2it\theta(\frac{1}{\nu_{j}})}}{z-\frac{1}{\nu_{j}}}\\
      0 & 1 \\
        \end{array}
      \right), ~~&|z-\frac{1}{\nu_{j}}|=\rho,\quad j\in\bigtriangledown\setminus\Lambda,\\
      &\left(
        \begin{array}{cc}
      1 &  0 \\
      -\frac{z-\frac{1}{\nu_{j}}}{d_{j}e^{-2it\theta(\frac{1}{\nu_{j}})}} & 1 \\
        \end{array}
      \right), ~~&|z-\frac{1}{\nu_{j}}|=\rho, \quad j\in\bigtriangleup\setminus\Lambda.
\end{aligned}
\right.
\end{align}
where $k=1,2$.
\end{itemize}
\end{RHP}

For the jump condition \eqref{V2-tilde-jump}, we have the following Proposition.
\begin{prop}\label{V2-Estimate-2}
As $t\to\infty$, there exists positive constant $c$ such that
\begin{align}
||V^{(2)}-\mathbb{I}||_{L^{\infty}(\Sigma^{(o)})}
=O(e^{-ct}),
\end{align}
where $\Sigma^{(o)}=\left(\Sigma^{(2)}\cup\{|z-\nu_j|=\rho,|z-\frac{1}{\nu_j}|=\rho\}\right)\setminus\mathcal U_{\xi}$.
\end{prop}
\begin{proof}
The thinking method to prove the results in Proposition \ref{V2-Estimate-2} is similar to that in \cite{AIHP}. So we omit it.
\end{proof}

Proposition \ref{V2-Estimate-2} implies that the jump condition on $M^{out}(z)$ can be  completely ignored, because there is only exponentially small error (in $t$). Then, we decompose $M^{out}(z)$ as
\begin{align}\label{Trans-3}
M^{out}(z)=M^{err}(z)M^{out}_{\Lambda}(z),
\end{align}
where $M^{err}(z)$ is a error function and $M^{out}_{\Lambda}(z)$ solves RH problem \ref{RH-4} with $V^{(2)}(z)\equiv\mathbb{I}$. Additionally, $M^{err}(z)$ is a solution of a small-norm RH problem. Then, RH problem \ref{RH-4} is reduced to the following RH problem.

\begin{RHP}\label{RH-6}
Find a matrix value function $M^{out}_{\Lambda}(z)$, admitting
\begin{itemize}
 \item $M^{out}_{\Lambda}(z)$ is continuous in $\mathbb{C}\setminus\{\nu_k,\frac{1}{\nu_k}\}_{k\in\Lambda}$.
 \item $M^{out}_{\Lambda}(z)= \mathbb{I}+O(z^{-1})$ as $z\rightarrow \infty$.
  \item $M^{out}_{\Lambda}(z)= -\frac{i}{z}\sigma_3Q_{\pm}+O(z)$ as $z\rightarrow 0$.
  \item  $M^{out}_{\Lambda}(z)$ admits the residue conditions at poles  $\nu_{k}$ and $\frac{1}{\nu_{k}}$ for $k\in\Lambda$ i.e.,\\
     For $k\in\bigtriangleup\cap\Lambda$,
\begin{align}\label{Mout-L-Res-1}
\begin{split}
&\mathop{Res}\limits_{z=\nu_{k}}M^{out}_{\Lambda}(z)=\lim_{z\rightarrow \nu_{k}}M^{out}_{\Lambda}(z)\left(\begin{array}{cc}
    0 & b_{k}^{-1}\left(T'(\nu_{k})\right)^{-2}e^{-2it\theta(\nu_k)}\\
    0 & 0 \\
  \end{array}
\right),\\
&\mathop{Res}\limits_{z=\frac{1}{\nu_{k}}}M^{out}_{\Lambda}(z)=\lim_{z\rightarrow \frac{1}{\nu_{k}}}M^{out}_{\Lambda}(z)\left(\begin{array}{cc}
    0 & 0\\
    d_{k}^{-1}\left(\frac{1}{T}'(\frac{1}{\nu_{k}})\right)^{-2}e^{ 2it\theta(\frac{1}{\nu_{k}})} & 0 \\
  \end{array}
\right);
\end{split}
\end{align}
   For $k\in\bigtriangledown\cap\Lambda$,
   \begin{align}\label{Mout-L-Res-2}
\begin{split}
&\mathop{Res}\limits_{z=\nu_{k}}M^{out}_{\Lambda}(z)=\lim_{z\rightarrow \nu_{k}}M^{out}_{\Lambda}(z)\left(\begin{array}{cc}
    0 & 0\\
    b_{k}\left(T(\nu_{k})\right)^{-2}e^{2it\theta(\nu_k)} & 0 \\
  \end{array}
\right),\\
&\mathop{Res}\limits_{z=\frac{1}{\nu_{k}}}M^{out}_{\Lambda}(z)=\lim_{z\rightarrow \frac{1}{\nu_{k}}}M^{out}_{\Lambda}(z)\left(\begin{array}{cc}
    0 & d_{k}T^{2}\left(\frac{1}{\nu_{k}}\right)e^{-2it\theta(\frac{1}{\nu_{k}})}\\
    0 & 0 \\
  \end{array}
\right).
\end{split}
\end{align}
\end{itemize}
\end{RHP}

\begin{prop}\label{Msol-prop}
The RH problem \ref{RH-6} possesses unique solution. Moreover, $M^{out}_{\Lambda}(x,t;z)$ has  equivalent solution to the original RH problem \ref{RH-1} with modified scattering data $D_{\Lambda}=\{r(z)\equiv0, \check r(z)\equiv0, \{\nu_k,\frac{1}{\nu_k}, \tilde{b}_k, \tilde{d}_k\}_{k\in\Lambda}\}$ under the condition that $r(z)\equiv0$ and $\check r(z)\equiv0$ as follows.
\begin{enumerate}[I.]
  \item If $\Lambda=\varnothing$, then
  \begin{align}
    M^{out}_{\Lambda}(x,t;z)=\mathbb I +\frac{\sigma_2}{z}.
  \end{align}
  \item If $\Lambda\neq\varnothing$, assuming that there exist $\tilde{\mathcal N}$ discrete spectral points belonging to $\Lambda$, i.e., $\Lambda=\{j_{1},j_{2},\ldots,j_{\tilde{\mathcal N}}\}$, then
      \begin{align}
   M^{out}_{\Lambda}(z)=\mathbb{I}+\frac{\sigma_2}{z}+\sum^{\tilde{\mathcal N}}_{k=1}\left(
   \begin{array}{cc}
   \frac{\alpha_{k}}{z-\nu_{j_{k}}} & \frac{\beta^{*}_{k}}{z-\frac{1}{\nu_{j_{k}}}} \\
   \frac{\beta_{k}}{z-\nu_{j_{k}}} & -\frac{\alpha^*_{k}}{z-\frac{1}{\nu_{j_{k}}}}  \\
    \end{array}
    \right).\label{Mout-Lambda-ext}
  \end{align}
  where $\alpha_{k}=\alpha_{k}(x,t)$ and $\beta_{k}=\beta_{k}(x,t)$ with linearly dependant equations:
  \begin{align*}
    &b_{j_{s}}^{-1}e^{-2it\theta(\nu_{j_{s}})}T^{2}(\nu_{j_{s}})\alpha_{j}=- \sum_{k=1}^{\tilde{\mathcal N}}\frac{\beta_{k}^*}{\nu_{j_{s}}-\frac{1}{\nu_{j_{k}}}},\\
    &b_{j_{s}}^{-1}e^{-2it\theta(\nu_{j_{s}})}T^{2}(\nu_{j_{s}})\beta_{j}= 1+\sum_{k=1}^{\tilde{\mathcal N}}\frac{-\alpha_{k}^*}{\nu_{j_{s}}-\frac{1}{\nu_{j_{k}}}},
  \end{align*}
  for $j=1,2,\ldots,\tilde{\mathcal N}$.
\end{enumerate}
\end{prop}

\begin{proof}
According to the Liouville's theorem, the uniqueness of solution follows  immediately. For case I, it is obvious to obtain. As for Case II, based on the symmetries shown in Proposition \ref{pp2}, we can show that $M^{out}_{\Lambda}(z)$ admits a partial
fraction expansion of following form as above. And in order to obtain  $\alpha_{k}(x,t)$ and $\beta_{k}(x,t)$, we
substitute \eqref{Mout-Lambda-ext} into \eqref{Mout-L-Res-2} and obtain  linearly dependant equations set above.
\end{proof}

Then, applying the above results, we obtain the following Corollary.
\begin{cor}\label{corollary-q-Lambda}
When reflection coefficient $r(z)=0$ and the transmission coefficient $\check r(z)=0$, the scattering matrix $S(z)$ becomes identity matrix. Denote $q_{\Lambda}(x,t,\tilde D_{\Lambda})$ is the
$N(\Lambda)$-soliton with scattering data $\tilde D_{\Lambda}=\left\{0,0,\{\nu_{k},\frac{1}{\nu_k}, b_kT^{-2},d_kT^{2}\}_{k\in\Lambda}\right\}$.
Based on formula \eqref{q-reconstruction-potential}, the solution $q_{\Lambda}(x,t,\tilde D_{\Lambda})$ of \eqref{e2} with scattering data $\tilde D_{\Lambda}$ is given by:
\begin{align}\label{q-L-reconstruction-potential}
  q_{\Lambda}(x,t,\tilde D_{\Lambda})=i\lim_{z\to\infty}(zM^{out}_{\Lambda}(x,t,z))_{12}.
\end{align}
Then, for case I,
\begin{align}
  q_{\Lambda}(x,t,\tilde D_{\Lambda})=1.
\end{align}
For case II,
\begin{align}
  q_{\Lambda}(x,t,\tilde D_{\Lambda})=1+i\lim_{z\to\infty}(zM^{out}_{\Lambda}(x,t,z))_{12} =i\sum_{k=1}^{\tilde{\mathcal N}}\beta^*_k.
\end{align}
\end{cor}

\subsection{The error function $M^{err}(z)$ between $M^{out}$ and $M^{out}_{\Lambda}$}\label{error-function-Merr}

In what follows, we are going  to study the error matrix-function $M^{err}(z)$. We first give the proof that the
error function $M^{err}(z)$ solves a small norm RH problem. Then, we show that  $M^{err}(z)$ can be expanded
asymptotically for large times. According to the decomposition \eqref{Trans-3}, we can derive a RH problem  with respect  to matrix function $M^{err}(z)$.

\begin{RHP}\label{RH-7}
Find a matrix-valued function $M^{err}(z)$ satisfies that
\begin{itemize}
 \item $M^{err}(z)$ is continuous in $\mathbb{C}\backslash \left(\Sigma^{(2)}\cup\mathcal Z\cup\check{\mathcal Z}\right)$;
 \item $M^{err}(z)= \mathbb{I}+O(z^{-1})$ as $z\rightarrow \infty$.
  \item $M^{err}(z)= -\frac{i}{z}\sigma_3Q_{\pm}+O(z)$ as $z\rightarrow 0$.
 \item $M^{err}_+(z)=M^{err}_-(z)V^{(err)}(z)$, \quad $z\in\Sigma^{(2)}$, where
\end{itemize}
 \begin{align}\label{VE-Jump-B}
 V^{(err)}(z)=
 M^{out}_{\Lambda}(z)V^{(2)}(z)M^{out}_{\Lambda}(z)^{-1}.
 \end{align}
\end{RHP}

The jump matrix  $V^{(err)}(z)$ in RHP \ref{RH-7} satisfies the following uniformly estimation.
\begin{prop}\label{VE-estimate-prop}
The jump matrix  $V^{(err)}(z)$ satisfies
\begin{align}\label{VE-estimate}
\big|\big|V^{(err)}-\mathbb{I}\big|\big|_{L^{p}(\Sigma^{(2)})} =O(e^{-ct}).
\end{align}
\end{prop}
\begin{proof}
From the Proposition \ref{Msol-prop},  we learn that $M_{\Lambda}$ is bounded on $\Sigma^{(2)}$. Then, we have
\begin{align}\label{VE-estimate-proof}
||V^{(err)}-\mathbb{I}||_{L^{p}(\Sigma^{(2)})} =||V^{(2)}-\mathbb{I}||_{L^{p}(\Sigma^{(2)})}.
\end{align}
Furthermore, by using Proposition \ref{V2-Estimate-2}, the result \eqref{VE-estimate} is obtained directly.
\end{proof}

The Proposition \ref{VE-estimate-prop} is the basic condition to ensure that RH problem \ref{RH-7} can be established as a  small-norm RH problem. Therefore,
the existence and uniqueness of the solution of the RH problem \ref{RH-7} can be guaranteed  by using a small-norm RH problem
\cite{Deift-1994-2,Deift-2003}. Next, we give a briefly description of this process .

According to the Beals-Coifman theory, we evaluate the decomposition of the jump matrix $V^{(err)}$
\begin{align*}
V^{(err)}(z)=(b_{-})^{-1}b_{+}, ~~b_{-}=\mathbb{I}, ~~b_{+}=V^{(err)}(z).
\end{align*}
Then, we introduce some notations
\begin{align*}
(\omega_{e})_{-}=\mathbb{I}-b_{-},~~(\omega_{e})_{+}=b_{+}-\mathbb{I}, ~~\omega_{e}=(\omega_{e})_{+}+(\omega_{e})_{-}=V^{(err)}(z)-\mathbb{I}.
\end{align*}
Furthermore, we define  the integral operator $C_{\omega_{err}}(L^{\infty}(\Sigma^{(2)})\rightarrow L^{2}(\Sigma^{(2)}))$ as
\begin{align*}
C_{\omega_{err}}f(z)=C_{-}(f(\omega_{e})_{+})+C_{+} (f(\omega_{e})_{-})=C_{-}(f(V^{(err)}(z)-\mathbb{I})),
\end{align*}
where $C_{-}$ is the Cauchy projection operator
\begin{align}\label{Cauchy-opera}
C_{-}(f)(z)\lim_{z\rightarrow\Sigma_{-}^{(2)}}\int_{\Sigma^{(2)}}\frac{f(s)}{s-z}ds,
\end{align}
and $||C_{-}||_{L^{2}}$ is bounded. Then,  the solution of RH problem \ref{RH-7} can be derived   as
\begin{align}\label{E-solution}
M^{err}(z)=\mathbb{I}+\frac{1}{2\pi i}\int_{\Sigma^{(2)}}\frac{\mu_{e}(s) (V^{(err)}(s)-\mathbb{I})}{s-z}\,ds,
\end{align}
where $\mu_{e}\in L^2 (\Sigma^{(2)})$  is the solution of the following equation
\begin{align}\label{mu-equation}
(1-C_{\omega_{err}})\mu_{e}=\mathbb{I}.
\end{align}
Next,  according to the properties of the Cauchy projection operator $C_{-}$ and  Proposition \ref{VE-estimate-prop}, we obtain
\begin{align}
\begin{split}
\|C_{\mu_{e}}\|_{L^2(\Sigma^{(2)})}& \lesssim\|C_-\|_{L^2(\Sigma^{(2)})\rightarrow L^2(\Sigma^{(2)})}\|V^{(err)}-\mathbb{I}\|_{L^{\infty}
(\Sigma^{(2)})}\\
&\lesssim O(e^{-ct}),
\end{split}
\end{align}
which means that $1-C_{\omega_{err}}$ is invertible.
In addition,
\begin{align}\label{mu-tildeE-estimation}
||\mu_{e}||_{L^{2}(\Sigma^{(2)})}\lesssim \frac{\|C_{\omega_{e}}\|}{1-\|C_{\omega_{e}}\|}\lesssim O(e^{-ct}).
\end{align}
Hence, $\mu_{e}$ is existence and uniqueness. Therefore,  we can say that the solution of RH problem \ref{RH-7} $M^{err}(z)$ exists.

In order to reconstruct the solutions of the NNLS equation \eqref{e2}, it is necessary and critical  to study the asymptotic behavior of $M^{err}(z)$ as $z\rightarrow\infty$.

\begin{prop}
The $M^{err}(z)$ defined in \eqref{Trans-3} satisfies
\begin{align}\label{E-estimation}
|M^{err}(z)-\mathbb{I}|\lesssim O(e^{-ct}).
\end{align}
\end{prop}
\begin{proof}
Based on \eqref{VE-estimate}, \eqref{E-solution} and \eqref{mu-tildeE-estimation}, the formula \eqref{E-estimation} can be derived directly.
\end{proof}

Then, on the basis of the above results, we have the following corollary.
\begin{cor}
For  $\xi\in(1,\infty)$ and $|t|\gg1$, uniformly for $z\in\mathbb{C}$,  $M^{out}(x,t;z)$  is expressed as
\begin{align}\label{Mout-MoutLambda}
M^{out}(z)=M^{out}_{\Lambda}(z)\left(\mathbb{I}+O(e^{-ct})\right),
\end{align}
Meanwhile, for $z\rightarrow\infty$, the asymptotic extension $M^{out}(x,t;z)$ is expressed as
\begin{align}\label{Msol-extension}
M^{out}(z)=M^{out}_{\Lambda}(z)\left(\mathbb{I}+z^{-1}O(e^{-ct})+O(z^{-2})\right).
\end{align}
\end{cor}

\subsection{Local solvable model near phase point $z=\xi$}\label{section-local-model}

On the basis of  Proposition \ref{V2-esti}, we know that $V^{(2)}-I$ does not have a uniform estimate for large time near the phase point $z=\xi$. Therefore,  we need to continue our study near the stationary phase points  in this section.

Define  local jump contour as (see Fig. \ref{fig-5})
\begin{align*}
\Sigma^{loc}&=\Sigma^{(2)}\cap\mathcal U_{\xi}\triangleq\Sigma_1\cup\Sigma_2,\\
\Sigma_k&=\bigcup_{j=1}^4\Sigma_{kj}\cap\mathcal U_{\xi},\quad k=1,2.
\end{align*}

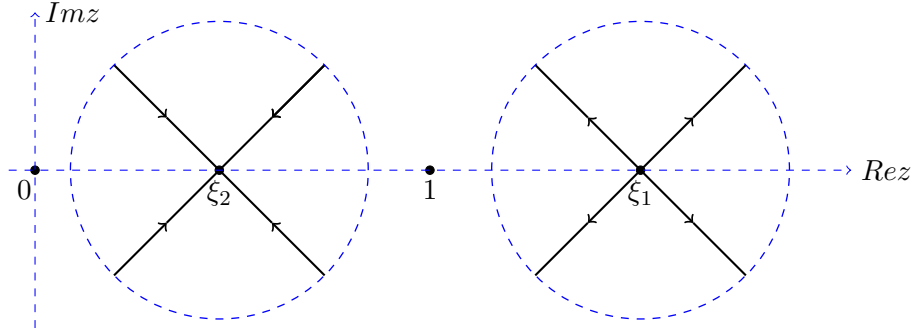
\begin{figure}
\centerline{\begin{tikzpicture}[scale=0.7]
\draw[->][thick](4,0)--(5,1);
\draw[-][thick](5,1)--(6,2);
\draw[->][thick](4,0)--(5,-1);
\draw[-][thick](5,-1)--(6,-2);
\draw[->][thick](4,0)--(3,1);
\draw[-][thick](3,1)--(2,2);
\draw[->][thick](4,0)--(3,-1);
\draw[-][thick](3,-1)--(2,-2);
\draw[->][thick](-2,2)--(-3,1);
\draw[-][thick](-3,1)--(-5,-1);
\draw[->][thick](-2,2)--(-3,1);
\draw[->][thick](-6,2)--(-5,1);
\draw[->][thick](-2,-2)--(-3,-1);
\draw[-][thick](-3,-1)--(-5,1);
\draw[->][thick](-6,-2)--(-5,-1);
\draw[fill] (4,0)node[below]{$\xi_{1}$} circle [radius=0.08];
\draw[fill] (-4,0)node[below]{$\xi_{2}$} circle [radius=0.08];
\draw[fill] (0,0)node[below]{$1$} circle [radius=0.08];
\draw[fill] (-7.5,0) circle [radius=0.08];
\draw[fill] (-7.7,0)node[below]{$0$};
\draw[fill] (8,0)node[right]{$Rez$};
\draw[fill] (-7.5,3)node[right]{$Imz$};
\draw[->][blue,dashed](-8,0)--(8,0);
\draw[->][blue,dashed](-7.5,-3)--(-7.5,3);
\draw(4,0) [dashed][blue, line width=0.5] circle(2.828);
\draw(-4,0) [dashed][blue, line width=0.5] circle(2.828);
\end{tikzpicture}}
  \caption{\small (Color online) The jump contour for the local model near the phase point $z=\xi$.}\label{fig-5}
\end{figure}

Then, according to the definition of  $\mathcal U_{\xi}$ we find that there are no discrete spectrum in $\mathcal{U}_{\xi}$. Therefore, we have $T(z)=\delta(z)$ and RH problem \ref{RH-rhp} can be reduced to the following model.

\begin{RHP}\label{RH-loc}
Find a matrix value function $M^{loc}$, admitting
\begin{itemize}
 \item $M^{loc,k}(x,t,z)$ is continuous in $\mathbb{C}\setminus(\Sigma^{loc})$.
 \item $M^{loc,k}(z)(x,t,z)= \mathbb{I}+O(z^{-1})$ as $z\rightarrow \infty$.
  \item $M^{loc,k}(x,t,z)= -\frac{i}{z}\sigma_3Q_{\pm}+O(z)$ as $z\rightarrow 0$.
 \item $M_+^{loc,k}(x,t,z)=M_{-}^{loc,k}(x,t,z)V^{loc}(z),$ ~~ $z\in\Sigma^{loc}$, where the jump matrix $V^{loc,k}(z)$ satisfies
  \begin{align}\label{V-loc-jump}
V^{loc,k}=\left\{\begin{aligned}
&\left(
  \begin{array}{cc}
    1 & -f_{kj}e^{-2it\theta(z)}  \\
    0 & 1 \\
  \end{array}
\right), ~~&z\in\Sigma_{kj}\cap\mathcal{U}_{\xi_k},~~j=1,3,\\
&\left(
  \begin{array}{cc}
    1 & 0 \\
    f_{kj}e^{2it\theta(z)} & 1 \\
  \end{array}
\right), ~~&z\in\Sigma_{kj}\cap\mathcal{U}_{\xi_k},~~j=2,4,
\end{aligned}
\right.
\end{align}
where $k=1,2$.
\end{itemize}
\end{RHP}

Since there are no discrete spectrum in $\mathcal{U}_{\xi}$, the local RH problem \ref{RH-loc} only has the jump condition on $\Sigma_1\cup\Sigma_2$ and has no poles. Next, we show that as $t\to\infty$, the interaction between the two local models on $\Sigma_1$ and $\Sigma_2$ reduces to zero,
and their contributions to the solution of $M^{loc}(x,t,z)$  is simply the
sum of the separate contributions from $\Sigma_1$ and $\Sigma_2$.

Define
\begin{align*}
  \omega_{-}=0,\quad \omega_+=V^{(2)}-\mathbb I,\quad \omega=V^{(2)}-\mathbb I,\\
  \omega=\omega_1+\omega_2,\quad \left\{\begin{aligned}\omega_1=0,\quad z\in\Sigma_2,\\
  \omega_2=0,\quad z\in\Sigma_1.
  \end{aligned}
\right.
\end{align*}

Then, based on the Cauchy projection operator shown in \eqref{Cauchy-opera}, we have
\begin{align*}
C_{\omega}f(z)=C_{-}(f\omega_{+})+C_{+} (f\omega_{-})=C_{-}(f(V^{(2)}(z)-\mathbb{I})),
\end{align*}
where $C_{\omega}=C_{\omega_1}+C_{\omega_1}$. Then, we have the following results.

\begin{prop}\label{C-interaction}
As $t\to\infty$,
\begin{align*}
\left\|C_{\omega_1}C_{\omega_2}\right\|_{L^{2}(\Sigma^{loc})}= \left\|C_{\omega_2}C_{\omega_1}\right\|_{L^{2}(\Sigma^{loc})}\lesssim t^{-1},\\
\left\|C_{\omega_1}C_{\omega_2}\right\|_{L^{2}(\Sigma^{loc})\to L^{2}(\Sigma^{loc})},\quad \left\|C_{\omega_2}C_{\omega_1}\right\|_{L^{2}(\Sigma^{loc})\to L^{2}(\Sigma^{loc})}\lesssim t^{-1}.
\end{align*}

\end{prop}

\begin{proof}
The thinking method to prove the results in Proposition \ref{C-interaction} is similar to that in \cite{Fan-dnls-nonzero}, so we omit it.
\end{proof}

On the basis of  the idea and step in \cite{Deift-1993},  the following result can be derived.
\begin{prop}
As $t\to\infty$,
\begin{align*}
 \int_{\Sigma^{loc}}\frac{(\mathbb I-C_{\omega})^{-1}\mathbb I \omega}{s-z}/,ds= \sum_{j=1}^2\int_{\Sigma_j}\frac{(\mathbb I-C_{\omega_j})^{-1}\mathbb I \omega_j}{s-z}/,ds+O(t^{-\frac{3}{2}}).
\end{align*}
\end{prop}

Therefore, as $s\to\infty$, the RH problem \ref{RH-loc} can be reduced to a model RH problem whose solution can be given explicitly in terms of parabolic cylinder functions on every contour  $\Sigma_1$ and $\Sigma_2$, respectively. In what follows, we are going to solve the problem near the phase points $z=\xi_1$ and $z=\xi_2$ by applying the parabolic cylinder(PC) model. Before we carry out this step, we first study the Taylor expansion of $\theta(z)$ in the neighborhood of $\xi_k(k=1,2)$.

\begin{align}\label{Taylor-theta}
\theta(z)=\theta(\xi_k)+\frac{\theta''(\xi_k)}{2}(z-\xi_k)^{2}+G_{k}(z;\xi_k).
\end{align}

\begin{prop}
Let $K$ be a sufficiently large constant. The operators $N_{\xi_k}$ mean
\begin{align*}
N_{\xi_k}:g(z)\to (N_{\xi_k}g)(z)=g\left(\xi_k+\frac{s}{\sqrt{2t\theta''(\xi_k)}\epsilon_k}\right),
\end{align*}
where $s=u\xi_ke^{\pm i\phi_k}$, $u<\rho$, $k=1,2$. Then, for $1<\frac{x}{2t}<K$, $G_{k}(z;\xi_k)$ satisfy
\begin{align}\label{G-asy}
\left|e^{-itG_{k}(\frac{s}{\sqrt{2t\theta''(\xi_k)}\epsilon_k};\xi_k)}\right|\to 1,\quad t\to \infty.
\end{align}
\end{prop}

Next, we first study this model problem near the phase points $\xi_1$.

\begin{itemize}
  \item \textbf{Near the phase point $z=\xi_1$}.
\end{itemize}

Recall that
\begin{align}\label{7-1}
\delta(z)^{-1}=\exp\left[i\int_{I(\xi)}\nu(s) \left(\frac{1}{s-z}-\frac{1}{2s}\right)ds\right]
=\frac{(z-\xi_1)^{i\nu(\xi_1)}}{(z+\xi_2)^{-i\nu(\xi_2)}}e^{\mathcal A(z)},
\end{align}
where $\mathcal A(z)=-\frac{1}{2\pi i}\left(\int_{I(\xi)}\log(z-s)d(\log(1-r(s)\check r(z)))- \int_{I(\xi)}\log(1-r(s)\check r(z))\frac{1}{2s}ds\right)$.
As $z\rightarrow \xi_1$,
\begin{align}\label{7-2}
\theta(z)=\theta(\xi_1)+\frac{\theta''(\xi_1)}{2}(z-\xi_1)^{2}+G_{1}(z;\xi_1).
\end{align}
We consider the following scaling transformation
\begin{align}\label{7-3}
(N_{\xi_1}f)(z)=f\left(\xi_1+\frac{z}{\sqrt{2t\theta''(\xi_1)}}\right),
\end{align}
then, we can derive that
\begin{align}\label{7-4}
(N_{\xi_1}\delta^{-1} e^{-it\theta(z)})(z)=\delta_{(\xi_1)}^{(0)}\delta_{(\xi_1)}^{(1)}(z),
\end{align}
where
\begin{gather*}
\delta_{(\xi_1)}^{(0)}=e^{-it\theta(\xi_1)}(\frac{1}{\sqrt{2t\theta''(\xi_1)}}) ^{-i\nu(\xi_1)}e^{\mathcal A(\xi_1)}e^{G_1\left(\xi_1+\frac{z}{\sqrt{2t\theta''(\xi_1)}};\xi_1\right)},\\
\delta_{(\xi_1)}^{(1)}(z)=z^{-i\nu(\xi_1)}
\left(\frac{z}{\sqrt{2t\theta''(\xi_1)}}+\xi_1-\xi_2\right)^{-i\nu(-z_{0})}
e^{\mathcal A\left(\xi_1+\frac{z}{\sqrt{2t\theta''(\xi_1)}}\right)-\mathcal A(\xi_1)}e^{-\frac{iz^{2}}{4}}.
\end{gather*}
From the expression of $\delta_{(z_{0})}^{(1)}(z)$, we can effortlessly obtain
that for $\zeta\in\{\zeta=u\xi_1e^{\pm i\phi}, -\rho<u<\rho\}$,
\begin{align}\label{7-5}
\delta_{(\xi_1)}^{(1)}(\zeta)\thicksim \zeta^{-i\nu(\xi_1)}e^{-\frac{i\zeta^{2}}{4}}, ~~ as~~ t\rightarrow +\infty,
\end{align}
which implies that the influence of the third power can be omitted. Thus, for large $t$, the solution of the Riemann-Hilbert problem for $M^{loc}(x,t;z)$ as formulated on the cross centered at $z=\xi_1$, can be approximated  based on the PC model(see Appendix A).

We introduce the transformation
\begin{align}\label{7-6}
\begin{split}
\lambda&=\lambda(\xi_1)=\sqrt{2t\theta''(z)}(z-\xi_1),\\
r_0=r_{0}^{\xi_1}&=r(\xi_1)T(\xi_1)^{-2}(\sqrt{2t\theta''(z)})^{-2i\nu(\xi_1)} e^{2it\theta(\xi_1)},
\end{split}
\end{align}
then, the solution $M^{loc,1}(x,t;z)$ formulated on the cross centered at $z=\xi_1$ can be obtained via applying the solution $M^{pc,1}(\lambda)$, as shown in Appendix $A$.
Then, the solution of $M^{loc,1}(x,t;z)$ at $z=\xi_1$ can be expressed as
\begin{align}\label{7-7}
M^{pc,1}(r_{0}^{\xi_1},\lambda)=\mathbb{I}+\frac{M_1^{pc,1}(\xi_1)}{i\lambda}+O(\lambda^{-2}),
\end{align}
where
\begin{align*}
M_1^{pc,1}=\begin{pmatrix}0&-\beta_{12}^{\xi_1}(r_{0}^{\xi_1})\\ \beta^{\xi_1}_{21}(r_{0}^{\xi_1})&0\end{pmatrix},
\end{align*}
with
\begin{align*}
\beta^{\xi_1}_{12}=\beta_{12}(r_{0}^{\xi_1})=
\frac{\sqrt{2\pi}e^{i\pi/4}e^{-\pi\nu/2}}{\frac{r_{0}^{\xi_1}} {1-r_{0}^{\xi_1}\check r_{0}^{\xi_1}}\Gamma(-i\nu)},\quad \beta_{21}^{\xi_1}=\beta_{21}(r_{0}^{\xi_1})=\frac{\nu}{\beta^{\xi_1}_{12}}.
\end{align*}

Furthermore, we consider the model problem near the phase points $\xi_2$.

\begin{itemize}
  \item \textbf{Near the phase point $z=\xi_2$}.
\end{itemize}

For $z\rightarrow \xi_2$,  we consider the scaling transformation
\begin{align}\label{7-8}
(N_{\xi_2}f)(z)=f\left(\xi_2+\frac{z}{\sqrt{-2t\theta''(\xi_2)}}\right),
\end{align}
then, we obtain
\begin{align}\label{7-9}
(N_{\xi_2}\delta^{-1} e^{-it\theta(z)})(z)=\delta_{(\xi_2)}^{(0)}\delta_{(\xi_2)}^{(1)}(z),
\end{align}
where
\begin{gather*}
\delta_{(\xi_2)}^{(0)}=e^{-it\theta(\xi_2)}(\frac{1}{\sqrt{-2t\theta''(\xi_2)}}) ^{-i\nu(\xi_2)}e^{\mathcal A(\xi_2)}e^{G_2\left(\xi_2+\frac{z}{\sqrt{-2t\theta''(\xi_2)}};\xi_2\right)},\\
\delta_{(\xi_2)}^{(1)}(z)=z^{i\nu(\xi_2)}
\left(\frac{z}{\sqrt{-2t\theta''(\xi_2)}}+\xi_2-\xi_1\right)^{-i\nu(\xi_2)}
e^{\mathcal A\left(\xi_2+\frac{z}{\sqrt{-2t\theta''(\xi_2)}}\right)-\mathcal A(\xi_2)}e^{\frac{iz^{2}}{4}}.
\end{gather*}
From the expression of $\delta_{(\xi_2)}^{(1)}(z)$, we can get the conclusion easily that for $\zeta\in\{\zeta=u\xi_2e^{\pm i\phi}, -\rho<u<\rho\}$,
\begin{align}\label{7-10}
\delta_{(\xi_2)}^{(1)}(\zeta)\thicksim (-\zeta)^{-i\nu(\xi_2)}e^{\frac{i\zeta^{2}}{4}}, ~~ as~~ t\rightarrow \infty,
\end{align}
which implies that the impact of the third power can be ignored. Thus, for large $t$, the solution of the RH problem for $M^{loc}(x,t;z)$ as formulated on the cross centered at $z=\xi_2$, can be approximated  based on the PC model.

We introduce the transformation
\begin{align}\label{7-11}
\begin{split}
\lambda&=\lambda(\xi_2)=\sqrt{-2t\theta''(z)}(z-\xi_2),\\
r_0=r_{0}^{\xi_2}&=r(\xi_2)T(\xi_2)^{-2}(\sqrt{-2t\theta''(z)})^{2i\nu(\xi_2)} e^{-2it\theta(\xi_2)},
\end{split}
\end{align}
then, the solution $M^{loc,2}(x,t;z)$ formulated on the cross centered at $z=\xi_2$ can be obtained via applying the solution $M^{pc,2}(\lambda)=\sigma M^{pc}(\lambda)\sigma$ shown in Appendix $A$, i.e.,
\begin{align}\label{7-12}
M^{pc,2}(r_{0}^{\xi_2},\lambda)=I+\frac{M_1^{pc,2}(\xi_2)}{i\lambda}+O(\lambda^{-2}),
\end{align}
where
\begin{align*}
M_1^{pc,2}(\xi_2)=\begin{pmatrix}0&\beta_{21}^{\xi_2}(r_{0}^{\xi_2})\\ -\beta^{\xi_2}_{12}(r_{0}^{\xi_2})&0\end{pmatrix},
\end{align*}
with
\begin{align*}
\beta^{\xi_2}_{12}=\beta_{12}(r_{0}^{\xi_2})=
\frac{\sqrt{2\pi}e^{i\pi/4}e^{-\pi\nu/2}}{r_{0}^{\xi_2}\Gamma(-i\nu)},\quad \beta_{21}^{\xi_2}=\beta_{21}(r_{0}^{\xi_2})=
\frac{-\sqrt{2\pi}e^{-i\pi/4}e^{-\pi\nu/2}}{(r_{0}^{\xi_2})^*\Gamma(i\nu)}
=\frac{\nu}{\beta^{\xi_2}_{12}}.
\end{align*}

For convenience,  we still use the notation $\lambda$ in the following analyses. Observing that $M^{loc,k}$ admits the asymptotic property
\begin{align}\label{7-13}
M^{loc,k}=\mathbb{I}+\frac{M_1^{pc,1}(\xi_1)}{i\lambda}+\frac{M_1^{pc,2}(\xi_2)}{i\lambda}+O(\lambda^{-2}),
\end{align}
we then substitute the first formula of \eqref{7-6} and \eqref{7-11} into \eqref{7-13}, and obtain
\begin{align}\label{7.14}
M^{loc,k}=\mathbb{I}+\frac{1}{\sqrt{2t\theta''(\xi_1)}}\frac{M_1^{pc,1}(\xi_1)}{z-\xi_1}
+\frac{1}{\frac{1}{\sqrt{-2t\theta''(\xi_2)}}}\frac{M_1^{pc,2}(\xi_2)}{z-\xi_2}+O(\lambda^{-2}).
\end{align}
In the local domain $\mathcal{U}_{\xi}$, we can obtain the result that
\begin{align}\label{Msp-Est}
|M^{loc,k}-\mathbb{I}|\lesssim O(t^{-\frac{1}{2}}), ~~as~~ t\rightarrow\infty,
\end{align}
which implies that
\begin{align}\label{7-15}
\|M^{loc,k}(z)\|_{\infty}\lesssim 1.
\end{align}
Since RH problem  and \ref{RH-3} possess the same jump conditions in $\mathcal{U}_{\xi}$, we apply $M^{loc,k}(z)$ to define a local model in two circles $z\in\mathcal{U}_{\xi}=\mathcal{U}_{\xi_1}\cup\mathcal{U}_{\xi_2}$.

\subsection{The small-norm RH problem for $E(z)$}
According to the transformation \eqref{Mrhp}, we have
\begin{align}\label{explict-Ez}
E(z)=\left\{\begin{aligned}
&M^{(2)}_{RHP}(z)M^{out}(z)^{-1}, &&z\in\mathbb{C}\setminus\mathcal{U}_{\xi},\\
&M^{(2)}_{RHP}(z)M^{loc,k}(z)^{-1}M^{out}(z)^{-1}, &&z\in\mathcal{U}_{\xi},
\end{aligned} \right.
\end{align}
which is analytic in $\mathbb{C}\setminus\Sigma^{(E)}$ where $\Sigma^{(E)}$ (see Fig. \ref{fig-7}) is defined as $$\Sigma^{(E)}=\partial\mathcal{U}_{\xi}\bigcup(\Sigma^{(2)} \setminus\mathcal{U}_{\xi}).$$

\begin{figure}
\centerline{\begin{tikzpicture}[scale=0.7]
\draw[green, ->][thick](4,0)--(6,2);
\draw[green,-][thick](6,2)--(7,3);
\draw[green,->][thick](4,0)--(6,-2);
\draw[green,-][thick](6,-2)--(7,-3);
\draw[green,->][thick](4,0)--(3,1);
\draw[green,-][thick](3,1)--(1.5,2.5);
\draw[green,->][thick](4,0)--(3,-1);
\draw[green,-][thick](3,-1)--(1.5,-2.5);
\draw[green,-][thick](1.5,2.5)--(1.5,-2.5);
\draw[green,->][thick](1.5,2.5)--(0,1);
\draw[green,-][thick](0,1)--(-3,-2);
\draw[green,->][thick](-3,2)--(-2,1);
\draw[green,->][thick](1.5,-2.5)--(0,-1);
\draw[green,-][thick](0,-1)--(-2,1);
\draw[green,->][thick](-3,-2)--(-2,-1);
\draw[green,-][thick](-3,2)--(-3,-2);
\draw[green,-][dashed](-3,2)--(-7,-2);
\draw[green,-][dashed](-3,-2)--(-7,2);
\filldraw[white, line width=0.5](-0.2,0) arc (0:360:0.8);
\filldraw[white, line width=0.5](4.8,0) arc (-360:0:0.8);
\draw [pink, dashed](-10,0)--(8,0);
\draw(4,0) [blue, line width=1] circle(0.8);
\draw(-1,0) [blue, line width=1] circle(0.8);
\draw[fill] (-5,0)node[below]{$0$} circle [radius=0.08];
\draw[fill] (4,0)node[below]{$\xi_1$} circle [radius=0.08];
\draw[fill] (-1,0)node[below]{$\xi_2$} circle [radius=0.08];
\draw[fill] (6,2)node[left]{$\Sigma^{(E)}$};
\draw[fill] (5.5,-0.1)node[above]{$\partial \mathcal {U}_{\xi_1}$};
\draw[fill] (-1,0.6)node[above]{$\partial \mathcal {U}_{\xi_2}$};
\end{tikzpicture}}
  \caption{\small (Color online) he jump contour $\Sigma^{(E)}=\partial\mathcal{U}_{\xi}\bigcup(\Sigma^{(2)}\setminus\mathcal{U}_{\xi})$ for the error function $E(z)$.}\label{fig-7}
\end{figure}
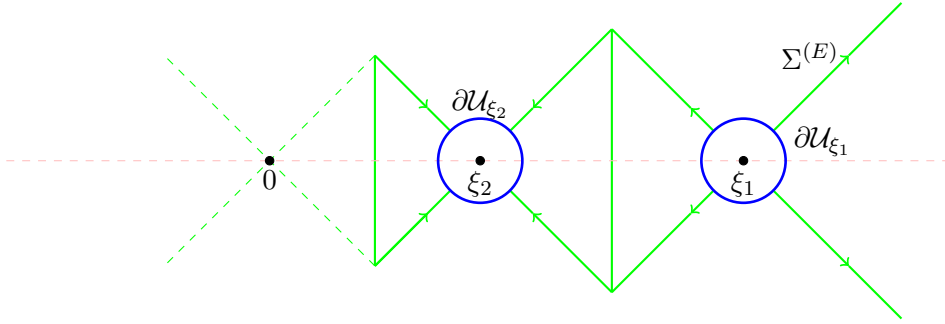

Then, we obtain a RH problem for  $E(z)$.
\begin{RHP}\label{RH-9}
Find a matrix-valued function $E(z)$ such that
\begin{itemize}
\item $E$ is analytic in $\mathbb{C}\setminus\Sigma^{(E)}$;
\item $E(z)=\mathbb{I}+O(z^{-1})$, \quad $z\rightarrow\infty$;
\item $E_+(z)=E_-(z)V^{(E)}(z)$, \quad $z\in\Sigma^{(E)}$, where
\end{itemize}
\begin{align}\label{7-17}
V^{(E)}(z)=\left\{\begin{aligned}
&M^{out}(z)V^{(2)}(z)M^{out}(z)^{-1}, &&z\in\Sigma^{(2)}\setminus \mathcal{U}_{\xi},\\
&M^{out}(z)M^{loc,k}(z)M^{out}(z)^{-1}, &&z\in\partial\mathcal{U}_{\xi}.
\end{aligned}\right.
\end{align}
\end{RHP}

Next, we evaluate the estimate of the jump matrix $V^{(E)}(z)$.

Based on Proposition \ref{V2-esti} and the boundedness of $M^{out}(z)$, as $t\to \infty$, we have
\begin{align}\label{VE-I-1}
|V^{(E)}(z)-\mathbb{I}|=\left\{\begin{aligned}
&O\left(e^{-c_1t}\right) &&z\in\Sigma^{(2)}\setminus\mathcal{U}_{\xi},\\
&O\left(e^{-c_2t}\right) &&z\in\Sigma',\\
&O\left(t^{-\frac{1}{2}}\right) &&z\in\partial\mathcal U_{\xi}.
\end{aligned}\right.
\end{align}

Then, by using a small-norm RH problem, the existence and uniqueness of RH problem \ref{RH-9} can be guaranteed.  Furthermore,  on the basis of Beals-Coifman theory, we obtain that
\begin{align}\label{Ez-solution}
E(z)=\mathbb{I}+\frac{1}{2\pi i}\int_{\Sigma^{(E)}}\frac{(\mathbb{I}+\mu_E(s))(V^{(E)}(s)-\mathbb{I})}{s-z}ds,
\end{align}
where $\mu_E\in L^2 (\Sigma^{(E)}) $ and satisfies
\begin{align}\label{7-18}
(1-C_{\omega_E})\mu_E=\mathbb{I},
\end{align}
where $C_{\omega_E}$ is an integral operator and  is defined as
\begin{align*}
C_{\omega_E}f=C_{-}\left(f(V^{(E)}-\mathbb{I})\right),\\
C_{-}f(z)=\lim_{z\rightarrow\Sigma_{-}^{(E)}}\frac{1}{2\pi i}\int_{\Sigma_E}\frac{f(s)}{s-z}ds,
\end{align*}
where $C_{-}$ is the Cauchy projection operator.
Next, based on the properties of the Cauchy projection operator $C_{-}$, and the estimate \eqref{VE-I-1}, we have
\begin{align}
\|C_{\omega_E}f\|_{L^2(\Sigma^{(E)})}\lesssim\|C_-\|_{L^2(\Sigma^{(E)})} \|f\|_{L^2(\Sigma^{(E)})}
\|V^{(E)}-\mathbb{I}\|_{L^{\infty}(\Sigma^{(E)})}\lesssim O(t^{-1/2}),
\end{align}
which implies $1-C_{\omega_E}$ is invertible. Then, the existence and uniqueness of $\mu_E$ is established. As a result,  the existence and uniqueness of $E(z)$ are guaranteed. These facts show  that the definition of $M^{(2)}_{RHP}$ is reasonable.

Furthermore, to reconstruct the solutions of $q(x,t)$, we need to study the asymptotic behavior of $E(z)$ as $z\rightarrow\infty$ and large time asymptotic behavior of $E(z)$. On the basis of \eqref{VE-I-1}, as $t\rightarrow+\infty$, we only need to consider the calculation on $\partial\mathcal{U}_{\xi}$ because it approaches to zero exponentially on other boundaries. Then, as $z\rightarrow \infty$, we can obtain that
\begin{align}\label{7-19}
E(z)=\mathbb I+\frac{E_{1}}{z}+O(z^{-2}),
\end{align}
where
\begin{gather}
E_{1}=-\frac{1}{2\pi i}\int_{\Sigma^{(E)}}(\mathbb{I}+\mu_E(s))(V^{(E)}(s)-I)/,ds.\label{7-22}
\end{gather}
Then, the large time, i.e., $t\rightarrow+\infty$, asymptotic behavior  of $E_{1}$  can be derived as
\begin{align}
\begin{split}
E_{1}=&-\frac{1}{2i\pi}\int_{\partial\mathcal{U}_{\xi}}(V^{(E)}(s)-I)ds+o(t^{-1})\\
=&\frac{1}{i\sqrt{2t\theta''(\xi_1)}}M^{out}(\xi_1)^{-1}M_1^{pc,1}(\xi_1)M^{out}(\xi_1)\\
&+\frac{1}{i\sqrt{-2t\theta''(\xi_2)}}M^{out}(\xi_2)^{-1}M_1^{pc,2}(\xi_2)M^{out}(\xi_2)
+\mathcal{O}(t^{-1}),
\end{split}\label{7-25}
\end{align}
where $M_1^{pc,k}(\xi_k)(k=1,2)$ are respectively defined in \eqref{7-12} and \eqref{7-7}.

\section{Analysis on pure $\bar{\partial}$-Problem}\label{section-Pure-dbar-RH}

In the above analysis, we have used a model RH problem $M^{(2)}_{RHP}$ to reduce $M^{(2)}$ to a pure $\bar{\partial}$-problem $M^{(3)}(z)$. Next, we are going  to investigate the existence and asymptotic behavior of the $\bar{\partial}$-RH problem \ref{delate-pure-RHP}.

According to Beals-Coifman theorem, we can use the following integral equation to express  the solution of the pure $\bar{\partial}$-RH problem \ref{delate-pure-RHP}, i.e.,
\begin{align}\label{8-1}
M^{(3)}(z)=\mathbb{I}-\frac{1}{\pi}\iint_{\mathbb{C}}\frac{M^{(3)}W^{(3)}}{s-z}\,dA(s),
\end{align}
where $A(s)$ is Lebesgue measure. Then, applying the Cauchy-Green integral operator, we have
\begin{align}\label{8-3}
f\cdot \mathrm{C}_z(z)=-\frac{1}{\pi}\iint_{\mathbb{C}}\frac{f(s)W^{(3)}(s)}{s-z}/,dA(s),
\end{align}
to rewrite the integral equation \eqref{8-1} as the following  operator form
\begin{align}\label{8-2}
M^{(3)}(z)=\mathbb{I}\cdot(\mathbb{I}-\mathrm{C}_z)^{-1}.
\end{align}

According to the expression of $M^{(3)}(z)$ \eqref{8-2}, we know that if the operator $(\mathbb{I}-\mathrm{C}_z)$ is  invertible, the solution $M^{(3)}(z)$ exists. Therefore, we next give the proof of the invertibility of the operator $\mathbb{I}-\mathrm{C}_z$.

\begin{prop}\label{J-estimate-prop}
For   $\xi\in(-1,1)$, as $t\to\infty$, there exists a constant $c=c(q_0,\xi)$ that enables the operator $\mathrm{C}_z: L^{\infty}(\mathbb{C})\rightarrow L^{\infty}(\mathbb{C}_z)\cap C^{0}(\mathbb{C})$ to satisfies
\begin{align}\label{8-4}
||\mathbb{C}_z||_{L^{\infty}(\mathbb{C})\rightarrow L^{\infty}(\mathbb{C})}\leq c|t|^{-\frac{1}{4}}.
\end{align}
\end{prop}

\begin{proof}
Without loss of generality, we just focus on the case that the matrix function supported in the region $\Omega_{21}$ for detail, the other cases can be proved in a similar way.  Firstly, we  assume that $f\in L^{\infty}(\Omega_{21})$, $s=u+iv$ and $z=x+iy$. Then, based on  \eqref{8-3}, we have
\begin{align}\label{J-estimate-proof-1}
|f\cdot\mathbb{C}_z(z)|\leq||f(s)||_{L^{\infty}(\mathbb{C})}\frac{1}{\pi} \iint_{\mathbb C}\frac{|W^{(3)}(s)|}
{|s-z|}dm(s),
\end{align}
where $W^{(3)}(s)=M^{(2)}_{RHP}(s)\bar{\partial}R^{(2)}(s)M^{(2)}_{RHP}(s)^{-1}$.
On the basis of Section \ref{section-pure RH problem}, we know  that $M^{(2)}_{RHP}(z)$ is bounded on $\Sigma^{(2)}$.  Therefore, by using the  formula \eqref{dbar-R2}, the inequality  \eqref{J-estimate-proof-1}  is reduced to
\begin{align}\label{8-5}
|f\cdot\mathbb{C}_z(z)|\lesssim \frac{1}{\pi}\iint_{\Omega_{21}}
\frac{|\bar{\partial}R_{21}(s)||e^{Re(-2it\theta(s))}}{|s-z|}|dm(s).
\end{align}

Then, according to the Proposition \ref{R-property} and the estimates shown in Appendix $A$, we have  the following result
\begin{align}\label{8-6}
||\mathbb{C}_z(z)||_{L^{\infty}\rightarrow L^{\infty}}\leq c(I_{1}+I_{2})\leq ct^{-\frac{1}{4}},
\end{align}
where
\begin{align}\label{8-8}
\begin{split}
I_{1}=\iint_{\Omega_{21}}
\frac{|s-\xi_2|^{-\frac{1}{4}}|e^{Re(-2it\theta(s))}|}{\sqrt{(u-x)^{2}+(v-y)^{2}}}\,dudv, \\
I_{2}=\iint_{\Omega_{21}}
\frac{|r'(|z|)||e^{Re(-2it\theta(s))}|}{\sqrt{(u-x)^{2}+(v-y)^{2}}}\,dudv.
\end{split}
\end{align}
\end{proof}

Proposition \ref{J-estimate-prop} implies that $(\mathbb{I}-\mathrm{C}_z)$ is  invertible as $t\to\infty$, i.e., $M^{(3)}(z)=\mathbb{I}+M^{(3)}\cdot\mathbb{C}_z(z)$.

Next, in order to reconstruct the long-time asymptotic behaviors of $q(x,t)$, we need to evaluate the asymptotic expansion of $M^{(3)}(z)$ as $z\rightarrow\infty$.
Based on \eqref{q-reconstruction-potential}, we need to determine the coefficient of the $z^{-1}$ term of $M^{(3)}$ in the Laurent expansion at infinity. Based on the equation \eqref{8-1}, we can derive that
\begin{align*}
M^{(3)}(z)=\mathbb{I}-\frac{1}{\pi}\iint_{\mathbb{C}}\frac{M^{(3)}(s)W^{(3)}(s)}
{s-z}\mathrm{d}A(s)
=\mathbb{I}+\frac{M^{(3)}_{1}}{z}+o(z^{-2}),
\end{align*}
where
\begin{align*}
M^{(3)}_{1}=\frac{1}{\pi}\iint_{\mathbb{C}}M^{(3)}(s)W^{(3)}(s)
\,dA(s).
\end{align*}
Next, we evaluate the asymptotic behavior of $M^{(3)}_{1}$ as $t\to\infty$.

\begin{prop}\label{prop-8-2}
For large $t$, there exists a constant $c$ that makes $M^{(3)}_{1}$ admits the following inequality
\begin{align}\label{8-7}
M^{(3)}_{1}\leq ct^{-\frac{3}{4}}.
\end{align}
\end{prop}

The proof of this Proposition is similar to Appendix $B$.

\section{Asymptotic approximation for the NNLS equation}

Now, we are going to construct the long time asymptotic of the NNLS equation \eqref{e2}.
Recall a series of transformations including \eqref{Trans-1}, \eqref{Trans-2}, \eqref{delate-pure-RHP} and \eqref{Mrhp}, i.e.,
\begin{align*}
M(z)\leftrightarrows M^{(1)}(z)\leftrightarrows M^{(2)}(z)\leftrightarrows M^{(3)}(z) \leftrightarrows E(z),
\end{align*}
we then obtain
\begin{align*}
M(z)=M^{(3)}(z)E(z)M^{out}(z)R^{(2)^{-1}}(z)T^{-\sigma_{3}}(z),~~ z\in\mathbb{C}\setminus\mathcal{U}_{\xi}.
\end{align*}
In order to recover the solution $q(x,t)$ , we take $z\rightarrow\infty$ along the imaginary axis, which implies $z\in\Omega_{up}$ or $z\in\Omega_{down}$, thus $R^{(2)}(z)=\mathbb I$. Then, we obtain
\begin{align*}
M=&\mathbb{I}+\frac{{M}_{1}}{z}+\cdots\\
=&\left(\mathbb{I}+\frac{M^{(3)}_{1}}{z}+\cdots\right)\left(\mathbb{I}+\frac{E_{1}}{z}+\cdots\right)
\left(\mathbb{I}+\frac{M_{1}^{(out)}}{z}+\cdots\right)\left(\mathbb{I}
+\frac{T_{1}\sigma_{3}}{z}+\cdots\right),
\end{align*}
from which we can derive that
\begin{align*}
M_{1}=M_{1}^{out}+E_{1}+M_{1}^{(3)}+T_{1}\sigma_{3}.
\end{align*}

Then, according to the reconstruction formula \eqref{q-reconstruction-potential}  and Proposition \ref{prop-8-2}, as $t\rightarrow\infty$, we obtain that
\begin{align}\label{9-1}
q(x,t)=i(M_{1}^{out})_{12} +i(E_{1})_{12}+O(t^{-\frac{3}{4}}).
\end{align}
Based on the Corollary \ref{corollary-q-Lambda} and formulae \eqref{7-25},  we obtain the soliton resolution
\begin{align*}
q(x,t)=q_{\Lambda}(x,t,\tilde D_{\Lambda}) +t^{-\frac{1}{2}}f_{\xi_1}+t^{-\frac{1}{2}}f_{\xi_2}+O(t^{-\frac{3}{4}}).
\end{align*}
where $q_{\Lambda}(x,t,\tilde D_{\Lambda})$ is defined in Corollary \ref{corollary-q-Lambda} and
\begin{align}\label{f-xi}
f_{\xi_1}=&\frac{1}{i\sqrt{2\theta''(\xi_1)}}M^{out}(\xi_1)^{-1} M_1^{pc,1}(\xi_1)M^{out}(\xi_1),\\
f_{\xi_2}=&\frac{1}{i\sqrt{-2t\theta''(\xi_2)}}M^{out}(\xi_2)^{-1} M_1^{pc,2}(\xi_2)M^{out}(\xi_2),
\end{align}
where $M^{out}(z)$ is defined in \eqref{Mrhp}, and $M_{1}^{pc,k}(\xi_k)$ is respectively defined in \eqref{7-7} and \eqref{7-12}.

The long time asymptotic behavior \eqref{9-1} gives the soliton resolution for the initial value problem of the NNLS equation. The soliton resolution contains the soliton term confirmed by $N(I)$-soliton on discrete spectrum and the $t^{-\frac{1}{2}}$ order term on continuous spectrum with residual error up to $O(t^{-\frac{3}{4}})$. Also, our results state that the soliton solutions of NNLS equation are asymptotic stable.

\begin{rem}
The steps in the steepest descent analysis of RHP \ref{RH-1} for $t\rightarrow-\infty$ is similar to the case  $t\rightarrow+\infty$, which has been presented in sections $5$-$9$. When we consider $t\rightarrow-\infty$, the main difference can be traced back to the fact that the regions of growth and decay of the exponential factors $e^{2it\theta}$ are reversed. Here, we leave the detailed calculations to the interested readers.
\end{rem}

Finally, we can give the results shown in Theorem \ref{Thm-1}

\begin{thm}\label{Thm-1}
Suppose that the initial value $q_{0}(x)$ satisfies the Assumption \eqref{assum} and $q_{0}(x)\in \mathcal{H}(\mathbb{R})$. The scattering data is denoted as $\left\{\nu_k,\frac{1}{\nu_k},r(z),\check r(z)\right\}_{k=1,2,\ldots,2N_1+2N_2}$  generated from the initial values $q_{0}(x)$. Let $q(x,t)$ be the solution of NNLS equation \eqref{e2}.Then as $t\rightarrow \infty$, the solution $q(x,t)$ can be expressed as
\begin{align}\label{9-2}
q(x,t)=q_{\Lambda}(x,t,\tilde D_{\Lambda}) +t^{-\frac{1}{2}}f_{\xi_1}+t^{-\frac{1}{2}}f_{\xi_2}+O(t^{-\frac{3}{4}}).
\end{align}
Here, $q_{\Lambda}(x,t,\tilde D_{\Lambda})$ is the $N(\Lambda)$ soliton solution, $f_{\xi_k}(k=1,2)$ is defined in $\eqref{f-xi}$.
\end{thm}

\begin{rem}
Theorem \ref{Thm-1} needs the condition $q_{0}(x)\in \mathcal{H}(\mathbb{R})$ so that the inverse scattering transform possesses  well mapping properties \cite{r-bijectivity}. Indeed, 
the asymptotic results only depend on the $H^{1,1}(\mathbb{R})$ norm of $r$ and $\check r$ in this work. So we restrict the initial potential $q_{0}(x)\in \mathcal{H}(\mathbb{R})$. Particularly, for any $q_{0}(x)\in \mathcal{H}(\mathbb{R})$ admitting the Assumption \ref{assum}, the process of the long-time analysis and calculations shown in this work is unchanged. In addition,
for  data in any weighted Sobolev space $H^{j,k}(\mathbb{R})$, there may exist spectral singularities. However, if the initial data decays exponentially, i.e., for any $c>0$, $\int_{R}e^{c|x|}|q_{0}(x)|dx<\infty$, then it is easy to know 
that the discrete spectrum cannot accumulate on the real axis. Isolated spectral singularities may still occur \cite{AIHP}.
\end{rem}

\section*{Acknowledgments}

The authors would like to thank Professor Engui Fan for his valuable  discussion and help.
This work was supported by  the National Natural Science Foundation of China under Grant No. 11975306, the Natural Science Foundation of Jiangsu Province under Grant No. BK20181351, the Six Talent Peaks Project in Jiangsu Province under Grant No. JY-059, the  333 Project in Jiangsu Province, and the Fundamental Research Fund for the Central Universities under the Grant Nos. 2019ZDPY07 and 2019QNA35.



\appendix
\section{The parabolic cylinder model problem}
Here, we describe the solution of  parabolic cylinder model problem\cite{PC-model,PC-model-2}.

Define the contour $\Sigma^{pc}=\cup_{j=1}^{4}\Sigma_{j}^{pc}$ where
\begin{align}
\Sigma_{j}^{pc}=\left\{\lambda\in\mathbb{C}|\arg\lambda=\frac{2j-1}{4}\pi \right\}.\tag{A.1}
\end{align}
For $r_{0}\in \mathbb{C}$, let $\nu(r)=-\frac{1}{2\pi}\log(1+|r_{0}|^{2})$, we consider the following parabolic cylinder model Riemann-Hilbert problem.
\begin{RHP}\label{PC-model}
Find a matrix-valued function $M^{(pc)}(\lambda)$ such that
\begin{align}
&\bullet \quad M^{(pc)}(\lambda)~ \text{is analytic in}~ \mathbb{C}\setminus\Sigma^{pc}, \tag{A.2}\\
&\bullet \quad M_{+}^{(pc)}(\lambda)=M_{-}^{(pc)}(\lambda)V^{(pc)}(\lambda),\quad
\lambda\in\Sigma^{pc}, \tag{A.3}\\
&\bullet \quad M^{(pc)}(\lambda)=\mathbb{I}+\frac{M_{1}}{\lambda}+O(\lambda^{2}),\quad
\lambda\rightarrow\infty, \tag{A.4}
\end{align}
where
\begin{align}\label{Vpc}
V^{(pc)}(\lambda)=\left\{\begin{aligned}
\lambda^{i\nu\hat{\sigma}_{3}}e^{-\frac{i\lambda^{2}}{4}
\hat{\sigma}_{3}}\left(
                    \begin{array}{cc}
                      1 & 0 \\
                      r_{0} & 1 \\
                    \end{array}
                  \right),\quad \lambda\in\Sigma_{1}^{pc},\\
\lambda^{i\nu\hat{\sigma}_{3}}e^{-\frac{i\lambda^{2}}{4}
\hat{\sigma}_{3}}\left(
                    \begin{array}{cc}
                      1 & \frac{r^{*}_{0}}{1+|r_{0}|^{2}} \\
                      0 & 1 \\
                    \end{array}
                  \right),\quad \lambda\in\Sigma_{2}^{pc},\\
\lambda^{i\nu\hat{\sigma}_{3}}e^{-\frac{i\lambda^{2}}{4}
\hat{\sigma}_{3}}\left(
                    \begin{array}{cc}
                      1 & 0\\
                      \frac{r_{0}}{1+|r_{0}|^{2}} & 1 \\
                    \end{array}
                  \right),\quad \lambda\in\Sigma_{3}^{pc},\\
\lambda^{i\nu\hat{\sigma}_{3}}e^{-\frac{i\lambda^{2}}{4}
\hat{\sigma}_{3}}\left(
                    \begin{array}{cc}
                      1 & r^{*}_{0} \\
                      0 & 1 \\
                    \end{array}
                  \right),\quad \lambda\in\Sigma_{4}^{pc}.
\end{aligned}\right.\tag{A.5}
\end{align}
\end{RHP}

\begin{figure}
\centerline{\begin{tikzpicture}[scale=0.6]
\draw[-][pink,dashed](-4,0)--(4,0);
\draw[-][thick](-4,-4)--(4,4);
\draw[-][thick](-4,4)--(4,-4);
\draw[->][thick](2,2)--(3,3);
\draw[->][thick](-4,4)--(-3,3);
\draw[->][thick](-4,-4)--(-3,-3);
\draw[->][thick](2,-2)--(3,-3);
\draw[fill] (3.2,3)node[below]{$\Sigma_{1}^{pc}$};
\draw[fill] (3.2,-3)node[above]{$\Sigma_{4}^{pc}$};
\draw[fill] (-3.2,3)node[below]{$\Sigma_{2}^{pc}$};
\draw[fill] (-2,-3)node[below]{$\Sigma_{3}^{pc}$};
\draw[fill] (0,0)node[below]{$0$};
\draw[fill] (1,0)node[below]{$\Omega_{6}$};
\draw[fill] (1,0)node[above]{$\Omega_{1}$};
\draw[fill] (0,-1)node[below]{$\Omega_{5}$};
\draw[fill] (0,1)node[above]{$\Omega_{2}$};
\draw[fill] (-1,0)node[below]{$\Omega_{4}$};
\draw[fill] (-1,0)node[above]{$\Omega_{3}$};
\draw[fill] (7,3)node[blue,below]{$\lambda^{i\nu\hat{\sigma}_{3}}e^{-\frac{i\lambda^{2}}{4}\hat{\sigma}_{3}}
\left(
  \begin{array}{cc}
    1 & 0 \\
    r_{0} & 1 \\
  \end{array}
\right)
$};
\draw[blue,fill] (7,-2)node[below]{$\lambda^{i\nu\hat{\sigma}_{3}}e^{-\frac{i\lambda^{2}}{4}\hat{\sigma}_{3}}
\left(
  \begin{array}{cc}
    1 & r^{*}_{0} \\
    0 & 1 \\
  \end{array}
\right)
$};
\draw[blue,fill] (-7,2.5)node[below]{$\lambda^{i\nu\hat{\sigma}_{3}}e^{-\frac{i\lambda^{2}}{4}\hat{\sigma}_{3}}
\left(
  \begin{array}{cc}
    1 & \frac{r^{*}_{0}}{1+|r_{0}|^{2}} \\
    0 & 1 \\
  \end{array}
\right)
$};
\draw[blue,fill] (-7,-1)node[below]{$\lambda^{i\nu\hat{\sigma}_{3}}e^{-\frac{i\lambda^{2}}{4}\hat{\sigma}_{3}}
\left(
  \begin{array}{cc}
    1 & 0 \\
    \frac{r_{0}}{1+|r_{0}|^{2}} & 1 \\
  \end{array}
\right)
$};
\end{tikzpicture}}
  \caption{\small (Color online) Jump matrix $V^{(pc)}$.}\label{fig-6}
\end{figure}
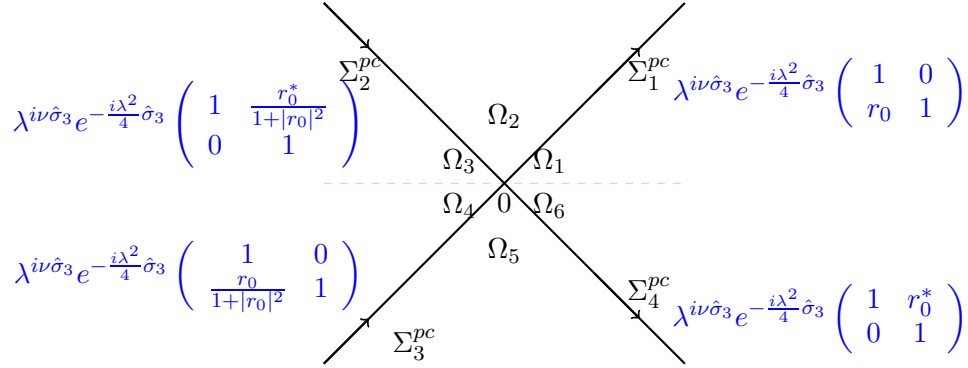

We know that the parabolic cylinder equation can be expressed as \cite{PC-equation}
\begin{align*}
\left(\frac{\partial^{2}}{\partial z^{2}}+(\frac{1}{2}-\frac{z^{2}}{2}+a)\right)D_{a}=0.
\end{align*}

As shown in literature \cite{Deift-1993, PC-solution2}, we obtain the explicit solution $M^{(pc)}(\lambda, r_{0})$:
\begin{align*}
M^{(pc)}(\lambda, r_{0})=\Phi(\lambda, r_{0})\mathcal{P}(\lambda, r_{0})e^{\frac{i}{4}\lambda^{2}\sigma_{3}}\lambda^{-i\nu\sigma_{3}},
\end{align*}
where
\begin{align*}
\mathcal{P}(\lambda, r_{0})=\left\{\begin{aligned}
&\left(
                    \begin{array}{cc}
                      1 & 0 \\
                      -r_{0} & 1 \\
                    \end{array}
                  \right),\quad &\lambda\in\Omega_{1},\\
&\left(
                    \begin{array}{cc}
                      1 & -\frac{r^{*}_{0}}{1+|r_{0}|^{2}} \\
                      0 & 1 \\
                    \end{array}
                  \right),\quad &\lambda\in\Omega_{3},\\
&\left(
                    \begin{array}{cc}
                      1 & 0\\
                      \frac{r_{0}}{1+|r_{0}|^{2}} & 1 \\
                    \end{array}
                  \right),\quad &\lambda\in\Omega_{4},\\
&\left(
                    \begin{array}{cc}
                      1 & r^{*}_{0} \\
                      0 & 1 \\
                    \end{array}
                  \right),\quad &\lambda\in\Omega_{6},\\
&~~~\mathbb{I},\quad &\lambda\in\Omega_{2}\cup\Omega_{5},
\end{aligned}\right.
\end{align*}
and
\begin{align*}
\Phi(\lambda, r_{0})=\left\{\begin{aligned}
\left(
                    \begin{array}{cc}
                      e^{-\frac{3\pi\nu}{4}}D_{i\nu}\left( e^{-\frac{3i\pi}{4}}\lambda\right) & -i\beta_{12}e^{-\frac{\pi}{4}(\nu-i)}D_{-i\nu-1}\left( e^{-\frac{i\pi}{4}}\lambda\right) \\
                      i\beta_{21}e^{-\frac{3\pi(\nu+i)}{4}}D_{i\nu-1}\left( e^{-\frac{3i\pi}{4}}\lambda\right) & e^{\frac{\pi\nu}{4}}D_{-i\nu}\left( e^{-\frac{i\pi}{4}}\lambda\right) \\
                    \end{array}
                  \right),\quad \lambda\in\mathbb{C}^{+},\\
\left(
                    \begin{array}{cc}
                      e^{\frac{\pi\nu}{4}}D_{i\nu}\left( e^{\frac{i\pi}{4}}\lambda\right) & -i\beta_{12}e^{-\frac{3\pi(\nu-i)}{4}}D_{-i\nu-1}\left( e^{\frac{3i\pi}{4}}\lambda\right) \\
                      i\beta_{21}e^{\frac{\pi}{4}(\nu+i)}D_{i\nu-1}\left( e^{\frac{i\pi}{4}}\lambda\right) & e^{-\frac{3\pi\nu}{4}}D_{-i\nu}\left( e^{\frac{3i\pi}{4}}\lambda\right) \\
                    \end{array}
                  \right),\quad \lambda\in\mathbb{C}^{-},
\end{aligned}\right.
\end{align*}
with
\begin{align*}
\beta_{12}=\frac{\sqrt{2\pi}e^{i\pi/4}e^{-\pi\nu/2}}{r_0\Gamma(-i\nu)},\quad \beta_{21}=\frac{-\sqrt{2\pi}e^{-i\pi/4}e^{-\pi\nu/2}}{r_0^*\Gamma(i\nu)}=\frac{\nu}{\beta_{12}},
\end{align*}
and $\Gamma$ denotes the gamma function.
Then, it is not hard to obtain the asymptotic behavior of the solution by using the well-known asymptotic behavior of $D_{a}(z)$,
\begin{align}\label{A-1}
M^{(pc)}(r_0,\lambda)=I+\frac{M_1^{(pc)}}{i\lambda}+O(\lambda^{-2}), \tag{A.6}
\end{align}
where
\begin{align*}
M_1^{(pc)}=\begin{pmatrix}0&\beta_{12}\\-\beta_{21}&0\end{pmatrix}.
\end{align*}

\section{Detailed calculations for the pure $\bar{\partial}$-Problem  }

\noindent \textbf{Proposition B.1}
For large $|t|$, there exist constants $c_{j}(j=1,2)$ such that $I_{j}(j=1,2)$,  defined in \eqref{8-8}, possess the following estimate
\begin{align}\label{B-1}
I_{j}\leq c_{j}t^{-\frac{1}{4}},~~ j=1,2. \tag{B.1}
\end{align}
\begin{proof}
Let $s=u+iv$ and $z=x+iy$. We first state a established estimation
\begin{align}\label{B-I2-1}
\Big|\Big|\frac{1}{s-z}\Big|\Big|^{2}_{L^{2}(0,\infty)}=(\int_{0}^{\infty} \frac{1}{|s-z|^{2}}dp)
\leq\frac{\pi}{|v-y|}.\tag{B.2}
\end{align}
Next, observing a fact that
\begin{align*}
 |e^{-2it\theta(z)}|\leqslant |e^{2t\theta''(\xi_2)(u-\xi_2)v}|,\quad t\to\infty,
\end{align*}
we rewrite $I_2$ as
\begin{align}\label{B-3}
\begin{split}
I_{2}\lesssim&\int_{\xi_2}^{+\infty}\int_{v+\xi_2}^{+\infty}
\frac{|r'(|z|)||e^{2t\theta''(\xi_2)(u-\xi_2)v}|}{\sqrt{(u-x)^{2}+(v-y)^{2}}}\,dudv\\
\lesssim&\int_{\xi_2}^{+\infty}\int_{v+\xi_2}^{+\infty}
\frac{|r'(|z|)||e^{2t\theta''(\xi_2)v^{2}}|}{\sqrt{(u-x)^{2}+(v-y)^{2}}}\,dudv\\
\lesssim&\int_{\xi_2}^{+\infty}|e^{2t\theta''(\xi_2)v^{2}}|\int_{v+\xi_2}^{+\infty}
\frac{|r'(|z|)|}{\sqrt{(u-x)^{2}+(v-y)^{2}}}\,dudv\\
\lesssim&\int_{\xi_2}^{+\infty}|e^{2t\theta''(\xi_2)v^{2}}|
\|r'(|z|)\|_{L^2}\left\|\frac{1}{|s-z|}\right\|\,dudv.
\end{split}\tag{B.3}
\end{align}

Since $r\in H^{1,1}(\mathbb{R})$, by applying \eqref{B-I2-1}, we have
\begin{align}\label{B-4}
\begin{split}
I_{2}\lesssim&\int_{0}^{+\infty}|e^{2t\theta''(\xi_2)v^{2}}| \frac{1}{\sqrt{|v-y|}}\,dudv\\
=&\int_{0}^{y}e^{2t\theta''(\xi_2)v^{2}}\frac{1}{\sqrt{y-v}}\,dv+ \int_{y}^{\infty}e^{2t\theta''(\xi_2)v^{2}}\frac{1}{\sqrt{v-y}}\,dv.
\end{split}\tag{B.4}
\end{align}

For the first term in \eqref{B-4}, taking $\omega:=\frac{v}{y}$, we obtain
\begin{align}\label{B-I2-4}
\begin{split}
\int_{0}^{y}\left(\frac{1}{y-v}\right)^{\frac{1}{2}} e^{2t\theta''(\xi_2)v^{2}}\,dv\\
\lesssim\int_{0}^{1}t^{-\frac{1}{4}}\frac{1}{\sqrt{\omega}\sqrt{1-\omega}}\,dv\lesssim t^{-\frac{1}{4}}.
\end{split}\tag{B.5}
\end{align}
Furthermore, for the remains item in \eqref{B-4}, taking $\omega=v-y$, we have
\begin{align}\label{B-I2-5}
\begin{split}
\int_{y}^{+\infty}\left(\frac{1}{v-y}\right)^{\frac{1}{2}} e^{2t\theta''(\xi_2)v^{2}}\,dv
&=\int_{0}^{+\infty}\frac{1}{\sqrt{\omega}}e^{2t\theta''(\xi_2)(\omega+y)^{2}}\,dv\\
&\lesssim \int_{0}^{+\infty}\frac{1}{\sqrt{\omega}}e^{2t\theta''(\xi_2)\omega^{2}}\,dv
\end{split}\tag{B.6}
\end{align}
Taking $\eta=\omega^{\frac{1}{2}}t^{\frac{1}{4}}$, from \eqref{B-I2-5}, we have
\begin{align}\label{B-I2-6}
\begin{split}
\int_{y}^{+\infty}\left(\frac{1}{v-y}\right)^{\frac{1}{2}} e^{2t\theta''(\xi_2)v^{2}}\,dv
&\lesssim 2t^{-\frac{1}{4}}\int_{0}^{+\infty} \frac{1}{\sqrt{\omega}}e^{2t\theta''(\xi_2)\eta^{4}}\,d\eta\\
&\lesssim t^{-\frac{1}{4}}.
\end{split}\tag{B.7}
\end{align}

Summarizing the above conclusions, we have $I_{2}\leq ct^{-\frac{1}{4}}$.

For $I_{1}$, by using the following inequality
\begin{align}\label{B-I1-1}
\Big|\Big|\frac{1}{|s-z|}\Big|\Big|_{L^{\eta}}\leq c|v-y|^{\frac{1}{\eta}-1},\tag{B.8}
\end{align}
we have
\begin{align}\label{B-I1-2}
\begin{split}
I_{1}\lesssim&\int_{\xi_2}^{+\infty}\int_{v}^{+\infty}
e^{2t\theta''(\xi_2)v^{2}}\frac{1}{|s-\xi_2|^{\frac{1}{2}}|s-z|}\,dudv\\
\lesssim&\int_{0}^{+\infty}e^{2t\theta''(\xi_2)v^{2}} \left\|\frac{1}{\sqrt{s-\xi_2}}\right\|_{L^{p}} \left\|\frac{1}{s-z}\right\|_{L^{q}}\,dv\\
=&\int_{0}^{y}e^{2t\theta''(\xi_2)v^{2}}v^{\frac{1}{p}-\frac{1}{2}} (y-v)^{\frac{1}{q}-1}\,dv+\int_{y}^{\infty}e^{2t\theta''(\xi_2)v^{2}}v^{\frac{1}{p}-\frac{1}{2}} (v-y)^{\frac{1}{q}-1}\,dv.
\end{split}\tag{B.9}
\end{align}
Then, similar to $I_{2}$ processing, we further obtain $I_{1}\leq ct^{-\frac{1}{4}}$.

Finally, we complete the proof.

\end{proof}

\bibliographystyle{plain}

\end{document}